\documentclass[10pt,conference,compsocconf]{IEEEtran}

\usepackage[utf8]{inputenc}
\usepackage{color}
\usepackage[usenames,dvipsnames]{xcolor}

\usepackage{centernot}
\usepackage{hhline}
\usepackage{mdframed}

\usepackage{amsthm}
\usepackage{bm} 

\usepackage{siunitx}
\usepackage{textcomp}
\usepackage{caption}
\usepackage{pst-func}
\usepackage{amssymb}
\usepackage{graphicx}
\usepackage{paralist}
\usepackage{environ}
\usepackage{pgfplots}
\usepackage{tikz}
\usetikzlibrary{shadows}
\usepackage{mathtools}
\usepackage{array,multirow}
\usepackage{algorithm}
\usepackage[hyphens]{url}
\usepackage{cite}
\usepackage{pgfplotstable}
\usepackage{threeparttable}

\usepackage{etoolbox}
\usepackage{arydshln}

\usepackage{xspace} 
 
\usepackage{dirtytalk} 
\usepackage{stmaryrd} 

\usepackage[inline,shortlabels]{enumitem} 
\newlist{inlinelist}{enumerate*}{1}
\setlist*[inlinelist,1]{%
  label=(\roman*),
}

\usepackage{xifthen}  
\usepackage{ifthen}

\usepackage{nicefrac}

\usepackage[hidelinks]{hyperref}
\usepackage{cleveref}
\usepackage{bigints}

\usetikzlibrary{pgfplots.dateplot}
\usetikzlibrary{matrix,chains,positioning,decorations.pathreplacing}
\usetikzlibrary{patterns,calc,fit,arrows}
\usetikzlibrary{shapes.geometric}

\usepackage[final,nomargin,inline,index]{fixme} 
\fxusetheme{color}

\FXRegisterAuthor{bart}{anbart}{\color{magenta} {\underline{bart}}}
\FXRegisterAuthor{tizi}{antizi}{\color{red} {\underline{tizi}}}
\FXRegisterAuthor{ste}{anste}{\color{blue} {\underline{ste}}}
\FXRegisterAuthor{nico}{annico}{\color{ForestGreen} {\underline{nico}}}
\FXRegisterAuthor{zun}{anzun}{\color{OliveGreen} {\underline{zun}}}

\hypersetup{
  breaklinks   = true,
  colorlinks   = true, 
  urlcolor     = blue, 
  linkcolor    = blue, 
  citecolor    = red   
}
\hypersetup{final} 

\usepackage{listings}
\definecolor{listingBG}{HTML}{FFFFCB}%
\definecolor{listingFrame}{HTML}{BBBB98}%
\definecolor{listingLineno}{rgb}{0.5,0.5,1.0}%

\definecolor{LightGrey}{rgb}{0.975,0.975,0.975}

\lstdefinelanguage{balzac}{
	commentstyle=\color{Gray},
	morecomment=[l]{//},
	morecomment=[s]{/*}{*/},
	classoffset=0,
        escapechar=\$,
	morekeywords={const,transaction,input,output,absLock,relLock,key,network,package},
	keywordstyle=\color{TealBlue}\bfseries,
	classoffset=1,
	morekeywords={sig,versig,ctxo,rtxo,ptxo,stxo,verscr,verctx,txid,fun,unit,int,string,bool,address,uint,date,checkDate,sha256},
	keywordstyle=\color{Blue}\bfseries,
	classoffset=2,
	morekeywords={BTC,true,false,and,or,if,then,else},
	keywordstyle=\color{Plum}\bfseries,
	classoffset=3,
	morekeywords={arg,val,wit,out,in,tkid,op,owner,tkval},
	keywordstyle=\color{MidnightBlue}\bfseries,
        frame=lines,
}

\newcommand{\ifempty}[3]{%
  \ifthenelse{\isempty{#1}}{#2}{#3}%
}

\newcommand{\ifdots}[3]{%
  \ifthenelse{\equal{#1}{...}}{#2}{#3}%
}

\newcommand{\hidden}[1]{}

\renewcommand{\vec}[1]{\boldsymbol{#1}}

\newcommand{\Real}[1]{\mathrm{Real}}

\newcommand{\codefont}{\fontsize{9}{9}\selectfont}

\newcommand{\lineno}[1]{{\tt\codefont {\textcolor{NavyBlue}{#1}}}}

\newcommand{\eg}{e.g.\@\xspace}
\newcommand{\ie}{i.e.\@\xspace}
\newcommand{\wrt}{w.r.t.\@\xspace}

\theoremstyle{plain}
\newtheorem{thm}{Theorem}
\newtheorem{lem}[thm]{Lemma}

\theoremstyle{definition}
\newtheorem{defn}{Definition}
\newtheorem{example}{Example}

\newenvironment{proofof}[2][]{%
  \ifempty{#1}
  {\subsection*{Proof of~\Cref{#2}}}
  {\subsection*{Proof of~\Cref{#2} ({#1})}}
  \label{#2-proof}
  }%
  {}

\newcommand{\sig}[3][]{\mathit{sig}^{#1}_{#2}\ifempty{#3}{}{({#3})}}

\newcommand{\txvername}{{\txColor{\sf ver}}}

\newcommand{\txver}[6]{\ifempty{#3}{\ifempty{#4}{\txvername_{#1}({#2},{#5},{#6})}
		{\txvername_{#1}^{{#3},{#4}}({#2},{#5},{#6})}}
	{\txvername_{#1}^{{#3},{#4}}({#2},{#5},{#6})}}

\newcommand{\BTC}{\textup{%
  \leavevmode
  \vtop{\offinterlineskip 
    \setbox0=\hbox{B}%
    \setbox2=\hbox to\wd0{\hfil\hskip-.03em
    \vrule height .3ex width .15ex\hskip .08em
    \vrule height .3ex width .15ex\hfil}
    \vbox{\copy2\box0}\box2}}\xspace}

\def\pmvColor{\color{ForestGreen}}
\newcommand{\pmvFmt}[1]{{\pmvColor{\sf #1}}}

\newcommand{\PartT}{\pmvFmt{{Hon}}\xspace} 

\newcommand{\pmv}[2][]{\pmvFmt{#2}_{\pmvColor{#1}}\xspace}
\newcommand{\pmvA}[1][]{\pmv[{#1}]{A}}
\newcommand{\pmvB}[1][]{\pmv[{#1}]{B}}
\newcommand{\pmvC}[1][]{\pmv[{#1}]{C}}
\newcommand{\pmvP}[1][]{\pmv[{#1}]{P}}
\newcommand{\Adv}{\pmv{Adv}} 

\def\txColor{\color{MidnightBlue}}
\def\fieldColor{\color{Plum}}
\newcommand{\txFmt}[1]{{\txColor{\sf #1}}}

\newcommand{\tx}[2][]{\txFmt{#2}_{\txColor{#1}}}
\newcommand{\txT}[1][]{\tx[#1]{T}}

\newcommand{\txTi}[1][]{\txFmt{T'_{\txColor{{\it #1}}}}}
\newcommand{\txTii}[1][]{\txFmt{T''_{\txColor{{\it #1}}}}}

\newcommand{\txTag}[3][]{{\fieldColor\sf #3}\ifempty{#1}{\ifempty{#2}{}{: {#2}}}{({#1})\ifempty{#2}{}{: {#2}}}}

\newcommand{\txIn}[2][]{\txTag[{#1}]{#2}{in}}
\newcommand{\txWit}[2][]{\txTag[{#1}]{#2}{wit}}
\newcommand{\txOut}[2][]{\txTag[{#1}]{#2}{out}}

\newcommand{\txf}{\txTag{}{f}} 
\newcommand{\txarg}{\txTag{}{arg}}
\newcommand{\txscript}{\txTag{}{scr}}
\newcommand{\txval}{\txTag{}{val}}

\newcommand{\tkop}{\txTag{}{op}}
\newcommand{\tkown}{\txTag{}{owner}}
\newcommand{\tkval}{\txTag{}{tkval}}
\newcommand{\tkid}{\txTag{}{tkid}}

\DeclareMathAlphabet{\mathbfsf}{\encodingdefault}{\sfdefault}{bx}{n}

\newcommand{\bcB}[1][]{{\mathbfsf{\txColor{B}}}_{\txColor{#1}}}

\newcommand{\eqdef}{\triangleq}

\newcommand{\mmid}{\,\|\,}

\newcommand{\irule}[2]{\dfrac{#1}{#2}}

\newcommand{\bnfdef}{::=}
\newcommand{\bnfmid}{\;|\;}

\newcommand{\nrule}[1]{{\scriptsize \textsc{#1}}}
\newcommand{\smallnrule}[1]{{\tiny \textsc{#1}}}
\newcommand{\sem}[2][]{\mbox{\ensuremath{\llbracket{#2}\rrbracket_{#1}}}}

\newcommand{\dom}[1]{\operatorname{dom} {#1}}
\newcommand{\ran}[1]{\operatorname{ran} {#1}}

\newcommand{\Nat}{\mathbb{N}}

\newcommand{\setenum}[1]{\{#1\}}
\newcommand{\setcomp}[2]{\left\{{#1} \,\middle|\, {#2}\right\}}

\newcommand{\seqat}[2]{{#1}.{#2}}
\newcommand{\runat}[2]{{#1}_{#2}}

\newcommand{\qedex}{\ensuremath{\diamond}}

\definecolor{LightGrey}{rgb}{0.95,0.95,0.95}
\definecolor{keyword}{HTML}{7F0055}

\newcommand{\tokScript}{\contrFmt{\expe_{\it TOK}}}
\newcommand{\btcScript}{\contrFmt{\expe_{\it BTC}}}

\def\tokColor{\color{MidnightBlue}}
\newcommand{\tokFmt}[1]{{\tokColor{#1}}}
\newcommand{\tok}[2][]{\tokFmt{#2}_{\tokColor{#1}}\xspace}
\newcommand{\tokT}[1][]{\tok[{#1}]{\tau}}
\newcommand{\tokTi}[1][]{\tok[{#1}]{\tau'}}

\newcommand{\genOp}{{\it gen}}
\newcommand{\burnOp}{{\it burn}}
\newcommand{\splitOp}{{\it split}}
\newcommand{\joinOp}{{\it join}}
\newcommand{\exchangeOp}{{\it xchg}}
\newcommand{\giveOp}{{\it give}}

\newcommand{\genRule}[1][]{\ifempty{#1}{\nrule{[Gen]}}{\smallnrule{[Gen]}}\xspace}
\newcommand{\burnRule}[1][]{\ifempty{#1}{\nrule{[Burn]}}{\smallnrule{[Burn]}}\xspace}
\newcommand{\splitRule}[1][]{\ifempty{#1}{\nrule{[Split]}}{\smallnrule{[Split]}}\xspace}
\newcommand{\joinRule}[1][]{\ifempty{#1}{\nrule{[Join]}}{\smallnrule{[Join]}}\xspace}
\newcommand{\exchangeRule}[1][]{\ifempty{#1}{\nrule{[Xchg]}}{\smallnrule{[Xchg]}}\xspace}
\newcommand{\giveRule}[1][]{\ifempty{#1}{\nrule{[Give]}}{\smallnrule{[Give]}}\xspace}

\newcommand{\actGen}[2]{\genOp({#1},{#2})}
\newcommand{\actBurn}[2]{\burnOp({#1},{#2})}
\newcommand{\actSplit}[3]{\splitOp({#1},{#2},{#3})}
\newcommand{\actJoin}[3]{\joinOp({#1},{#2},{#3})}
\newcommand{\actExchange}[2]{\exchangeOp({#1},{#2})}
\newcommand{\actGive}[2]{\giveOp({#1},{#2})}

\newcommand{\tokval}[3][]{\mathit{bal}^{#1}_{{#2}}\ifempty{#3}{}{({#3})}}
\newcommand{\genval}[2]{\mathit{minted}_{{#1}}\ifempty{#2}{}{({#2})}}
\newcommand{\burnval}[2]{\mathit{burnt}_{{#1}}\ifempty{#2}{}{({#2})}}

\newcommand{\tokvalC}[2]{\mathit{bal}_{{#1}}\ifempty{#2}{}{({#2})}}

\newlength\replength
\newcommand\repfrac{.1}

\newcommand\rulewidth{.6pt}
\newcommand\tdashfill[1][\repfrac]{\cleaders\hbox to \replength{%
  \smash{\rule[\arraystretch\ht\strutbox]{\repfrac\replength}{\rulewidth}}}\hfill}

\newcommand\tdotfill[1][\repfrac]{\cleaders\hbox to \replength{%
  \smash{\raisebox{\arraystretch\dimexpr\ht\strutbox-.1ex\relax}{.}}}\hfill}

\newcommand{\var}[2][]{#2_{#1}} 
\newcommand{\varX}[1][]{\var[#1]{x}} 
\renewcommand{\varXi}[1][]{\var[#1]{x'}} 
\newcommand{\varY}[1][]{\var[#1]{y}} 
\newcommand{\varYi}[1][]{\var[#1]{y'}} 
\newcommand{\varZ}[1][]{\var[#1]{z}}

\newcommand{\const}[2][]{#2_{#1}} 

\newcommand{\constPK}[1][]{\const[{#1}]{pk}}
\newcommand{\constSK}[1][]{\const[{#1}]{sk}}

\newcommand{\val}[2][]{#2_{#1}} 
\newcommand{\valV}[1][]{\val[#1]{v}}
\newcommand{\valVi}[1][]{\val[#1]{v'}}

\newcommand{\ArgA}[1][]{\vec{a}_{#1}}

\def\scriptColor{\color{Black}}
\newcommand{\script}[2][]{{\scriptColor{{\it #2}_{#1}}}}
\newcommand{\expe}[1][]{\script[{#1}]{e}}
\newcommand{\expei}[1][]{\script[{#1}]{e'}}

\newcommand{\versigName}{{\sf versig}}
\newcommand{\versig}[2]{\versigName({#1},{#2})}
\newcommand{\rtx}{{\sf rtx}}
\newcommand{\hashE}[1]{{\sf H}(#1)}
\newcommand{\hashSem}[1]{\ifempty{#1}{H}{H(#1)}}

\newcommand{\notE}[1]{{\sf not}~{#1}}

\newcommand{\sizeE}[1]{\ensuremath{| #1 |}}
\newcommand{\ifE}[3]{\mathsf{if}~{#1}~\mathsf{then}~{#2}~\mathsf{else}~{#3}}
\newcommand{\ifSem}[3]{\mathit{if}~{#1}~\mathit{then}~{#2}~\mathit{else}~{#3}}

\newcommand{\andE}{~{\sf and}~}
\newcommand{\orE}{~{\sf or}~}

\newcommand{\txout}{t} 
\newcommand{\txouti}{t'} 

\newcommand{\txo}{\fieldColor{\sf o}} 

\newcommand{\txof}[2]{{#1}.{#2}}
\newcommand{\ctxo}[1]{{\sf ctxo}\ifempty{#1}{}{.{#1}}}
\newcommand{\stxo}[2]{{\sf stxo}({#2})\ifempty{#1}{}{.{#1}}}
\newcommand{\rtxo}[2]{{\sf rtxo}({#2})\ifempty{#1}{}{.{#1}}}
\newcommand{\ptxo}[2]{{\sf ptxo}({#2})\ifempty{#1}{}{.{#1}}}
\newcommand{\txid}[1]{{\sf txid}({#1})}
\newcommand{\inidx}{{\sf inidx}}
\newcommand{\outidx}{{\sf outidx}}
\newcommand{\inlen}[1]{{\sf inlen({#1})}}
\newcommand{\outlen}[1]{{\sf outlen({#1})}}

\newcommand{\verscript}[2]{{\sf verscr}\ifempty{#1}{}{({#1},{#2})}}
\newcommand{\verrec}[1]{{\sf verrec}\ifempty{#1}{}{({#1})}}

\newcommand{\utxo}[1]{\mathit{UTXO}({#1})}

\newcommand{\sizeF}[1]{\mathit{size}\ifempty{#1}{}{({#1})}}
\newcommand{\true}{\mathit{true}}
\newcommand{\false}{\mathit{false}}

\newcommand{\CM}{\contrFmt{\expe_{\it CM}}}
\newcommand{\txr}[1]{{\fieldColor{\sf r}_{#1}}} 
\newcommand{\txp}{{\fieldColor{\sf p}}}

\def\contrColor{\color{RubineRed}}
\newcommand{\contrFmt}[1]{{\contrColor{\it #1}}}

\newcommand{\contrAdvC}[2]{\mathcal{C}} 

\newcommand{\splitname}{\textup{\texttt{split}}}

\newcommand{\splitB}[2]{{#1} \rightarrow {#2}}

\newcommand{\withdrawname}{\textup{\texttt{withdraw}}}
\newcommand{\withdrawC}[1]{\withdrawname\ifempty{#1}{}{\; {\pmv{#1}}}}

\newcommand{\afterName}{\texttt{after}}

\newcommand{\afterC}[2]{\textup{\afterName}\,{#1}\,\textup{\texttt{:}}\,{#2}}

\newcommand{\confG}[1][]{\Gamma_{#1}}
\newcommand{\confGi}[1][]{\Gamma'_{#1}}

\newcommand{\gnil}{\mathbf{0}}

\newcommand{\confDep}[3][]{\langle {#2}, {#3} \rangle_{#1}}
\newcommand{\confAuth}[3][]{{#2} \rhd_{#1} {#3}}

\newcommand{\labAuth}[3][]{{#2} \rhd_{#1} {\ifdots{#3}{\cdots}{#3}}} 

\newcommand{\labS}[1][]{\alpha_{#1}}        

\newcommand{\labC}{\gamma}          

\newcommand{\confS}[1]{\confG[{#1}]}

\newcommand{\runnameS}{\mathcal{S}}
\newcommand{\runnameC}{\mathcal{C}}

\newcommand{\runS}[1][]{\runnameS_{#1}}
\newcommand{\runSi}[1][]{\runnameS'_{#1}}
\newcommand{\runSii}[1][]{\runnameS''_{#1}}

\newcommand{\runC}[1][]{\runnameC_{#1}}
\newcommand{\runCi}[1][]{\runnameC'_{#1}}
\newcommand{\runCii}[1][]{\runnameC''_{#1}}

\newcommand{\stratS}[1]{\Sigma_{#1}^{\it s}}
\newcommand{\stratC}[1]{\Sigma_{#1}^{\it c}}

\newcommand{\stratSSet}{\mathbf{\Sigma}^{\it s}} 
\newcommand{\stratCSet}{\mathbf{\Sigma}^{\it c}} 

\newcommand{\stratMap}{\aleph}

\newcommand{\coher}[5]{\mathit{coher}({#1},{#2},{#3},{#4},{#5})}
\newcommand{\coherRel}[3]{{#1} \sim_{#3} {#2}}
\newcommand{\txMapC}{\mathit{txout}}
\newcommand{\tkMapC}{\mathit{tkid}}
\newcommand{\txMapCi}{\mathit{txout}'}
\newcommand{\tkMapCi}{\mathit{tkid}'}
\newcommand{\burnMapC}{\mathit{txburn}}
\newcommand{\burnMapCi}{\mathit{txburn}'}

\newenvironment{nscenter}
 {\parskip=0pt\par\nopagebreak\centering}
 {\parskip=2pt\par\noindent} 

\newtoggle{anonymous}
\togglefalse{anonymous}
\newtoggle{arxiv}
\toggletrue{arxiv}

\makeatletter
\renewcommand\paragraph{\@startsection{paragraph}{4}{\z@}%
  {2.25ex \@plus 1ex \@minus .2ex}%
  {-0.75em}%
  {\normalfont\normalsize\bfseries}}
\makeatother

\begin{document}

\title{Computationally sound Bitcoin tokens}

\iftoggle{anonymous}{
\author{
}
}{
\author{
  \IEEEauthorblockN{Massimo Bartoletti}
  \IEEEauthorblockA{\textit{Universit\`a degli Studi di Cagliari}}
  Cagliari, Italy \\                                                                
  \texttt{bart@unica.it}
  \and
  \IEEEauthorblockN{Stefano Lande}
  \IEEEauthorblockA{\textit{Universit\`a degli Studi di Cagliari}}
  Cagliari, Italy \\                                                                
  \texttt{lande@unica.it}
  \and
  \IEEEauthorblockN{Roberto Zunino}
  \IEEEauthorblockA{\textit{Universit\`a degli Studi di Trento}}
  Trento, Italy \\                                                                
  \texttt{roberto.zunino@unitn.it}
}
}

\maketitle

\begin{abstract}
  We propose a secure and efficient implementation of fungible tokens on Bitcoin.
  Our technique is based on a small extension of the Bitcoin script language,
  which allows the spending conditions in a transaction to depend on the neighbour transactions.
  We show that our implementation is computationally sound: 
  that is, adversaries can make tokens diverge from their ideal functionality only with negligible probability.
\end{abstract}


\section{Introduction}
\label{sec:intro}

One of the main applications of blockchain technologies is the exchange
of custom crypto-assets, called \emph{tokens}.
Token transfers currently involve \mbox{$\sim{50}\%$} of the 
transactions on the Ethereum blockchain~\cite{tokens},
and they are at the basis of many protocols built on top of that platform~\cite{Angelo20dapps,Frowis19fc}.
Broadly, tokens are classified as \emph{fungible} or \emph{non-fungible}.
Fungible tokens can be split into smaller units:
different units of the same token can be used interchangeably.
Further, users can join units of the same fungible token, 
and exchange them with other crypto-assets.
Instead, non-fungible tokens cannot be split or joined.

Historically, the first implementations of tokens were developed before Ethereum, on top of Bitcoin.
Some of them (\eg,~\cite{epobc}) used small bitcoin fractions to represent the token value;
some others (\eg, \cite{openAssets,colu,counterparty})
embedded the token value in other transaction fields~\cite{Bartoletti19grid}, 
to cope with the fluctuating bitcoin price. 
All these implementations have a common drawback:
the correctness of the token actions is \emph{not} guaranteed by the consensus protocol 
of the blockchain. 
In fact, the blockchain is used just to notarize the actions that manipulate tokens, 
but not to check that these actions are actually permitted. 
Typically, the owners of these tokens must resort to \emph{off-chain} mechanisms (\eg, trusted authorities) 
to have some guarantees on the correct use of tokens,
\eg that they are not double-spent, or that distinct tokens are not joined.

By contrast, modern blockchain platforms support \emph{on-chain} tokens,
whose correctness is guaranteed by the consensus protocol of the blockchain.
Some platforms (\eg, Algorand~\cite{algorandAssets}) natively support tokens, 
while some others (\eg, Ethereum) encode them as smart contracts.
Bitcoin, instead, does not support tokens natively, and its limited script language
is not expressive enough to implement them as smart contracts.
Since adding native tokens to Bitcoin appears to be out of reach, 
given the resilience of the Bitcoin community to radical changes~\cite{BIP0002},
the only viable alternative is to devise a small, 
efficient extension of the script language
which increases the expressiveness of Bitcoin enough to support tokens.

A recent Bitcoin Improvement Proposal (BIP 119~\cite{BIP119,Swambo20bitcoin})
aims at extending the Bitcoin script language with \emph{covenants},
a class of operators that allow a transaction 
to constrain how its funds can be used by the redeeming transactions.
Although these covenants enable a variety of use cases, 
\eg vaults, batched payments, and non-interactive payment channels~\cite{BIP119},
they are not expressive enough to implement 
fungible tokens in a practical way.
Roughly, the covenants proposed in the literature can
check that the token units are preserved upon $\splitOp$ actions, 
but they cannot ensure this property upon $\joinOp$ actions
(we describe these and other issues in~\Cref{sec:overview}).

In this work, we propose a variant of covenants,
named \emph{neighbourhood covenants},
which can inspect not only the redeeming transaction,
but also the siblings and the parent of the spent one. 
This extension preserves the basic UTXO design of Bitcoin,
adding only a few opcodes to its script language, which is
kept efficient, loop-free, and \emph{non} Turing-complete.
Still, neighbourhood covenants significantly increase the expressiveness of 
Bitcoin as a smart contracts platform, allowing to execute
arbitrary smart contracts by appending a \emph{chain} of transactions
to the blockchain. 
Technically, we prove that neighbourhood covenants make Bitcoin Turing-complete.

Although this expressiveness result is of theoretical interest,
in itself it does not enable an efficient implementation of tokens.
To recover efficiency, we implement token actions in a single, 
succinct script which exploits neighbourhood covenants.
We devote a large portion of the paper to establish the security of our construction: 
in brief, we define a symbolic model of token actions, 
and a computational model, 
where performing these actions corresponds to appending transactions to the Bitcoin blockchain.
Our main technical result is a \emph{computational soundness} theorem, 
which ensures that any execution in the computational model
has a corresponding execution in the symbolic one.
Therefore, we guarantee that a computational adversary 
cannot make the behaviour of tokens diverge from the behaviour
of the symbolic model.

\paragraph*{Contributions}

We summarise our contributions as follows:
\begin{itemize}

\item we introduce a symbolic model of fungible tokens,
which formalises their archetypal features:
their minting and burning, the split and join operations, 
and the exchange of tokens with other tokens or with bitcoins 
(\Cref{sec:tokens:s});
  
\item we propose neighbourhood covenants as a Bitcoin extension (\Cref{sec:covenants}), 
and we show that they make Bitcoin Turing-powerful (\Cref{th:covenants:turing-completeness}).
We then discuss how to efficiently implement them on Bitcoin; 

\item we exploit neighbourhood covenants to implement tokens on Bitcoin (\Cref{sec:tokens:c});

\item we introduce a computational model for Bitcoin and we prove the computational soundness of our token implementation (\Cref{th:computational-soundness} in~\Cref{sec:computational-soundness});

\item as an consequence, 
  we show that a value preservation property established in the symbolic model 
  can be lifted \emph{for free} to the computational model (\Cref{th:tokvals-tokvalc}).

\end{itemize}

\noindent
\iftoggle{arxiv}
{The proofs of our results are in the~\appendixname.}
{Due to space constraints, we provide the proofs of our statements in a separate technical report~\cite{BLZ20arxiv}.}

\section{Overview of the approach}
\label{sec:overview}

In this~\namecref{sec:overview} we summarize our approach:
in particular, we sketch our implementation of Bitcoin tokens,
motivating the use of neighbourhood covenants to guarantee their security.

\paragraph*{Tokens}

We propose a symbolic model of fungible tokens.
Since non-fungible tokens are the special case of fungible ones
where each token is generated exactly in one unit,
hereafter we consider the general case of fungible tokens.
The basic element of our model is the \emph{deposit},
\ie a term of the form:
\[
\confDep[\varX]{\pmvA}{\valV:\tokT}
\tag*{$(\valV \in \Nat)$}
\]
which represents the fact that 
a user $\pmvA$ owns $\valV$ units of a token $\tokT$,
where $\tokT$ may denote either user-defined tokens or bitcoins ($\BTC$).
The index $\varX$ uniquely identifies the term within a \emph{configuration},
\ie a composition of deposits, \eg:
\[
\confDep[\varX]{\pmvA}{1:\tokT}
\mid
\confDep[\varY]{\pmvA}{2:\tokT}
\mid
\confDep[\varZ]{\pmvB}{3:\BTC}
\]
We define a few actions to mint and manipulate tokens.
First, any user $\pmvA$ can mint $\valV$ units of a new token, 
spending a deposit of $0~\BTC$. 
Performing this action (say, with $\valV = 10$) 
is modelled as a state transition:
\begin{equation}
  \label{eq:overview:gen}
  \confDep[{\varX[0]}]{\pmvA}{0:\BTC} 
  \xrightarrow{\genOp} 
  \confDep[{\varX[1]}]{\pmvA}{10:\tokT} 
\end{equation}
The label $\genOp$ over the arrow records that the performed action
is a token minting.
The state transition~\eqref{eq:overview:gen} ensures that
the identifier $\varX[1]$ of the new deposit and
the identifier $\tokT$ of the minted token are \emph{fresh}.
After performing the action, $\pmvA$ owns a deposit of ten units of the token $\tokT$.
As said before, one of the peculiar properties of fungible tokens is that they can be $\splitOp$. 
When splitting her deposit in two smaller deposits, 
$\pmvA$ can choose the owner of one of the new deposits, \eg:
\begin{equation}
  \label{eq:overview:split}
  \confDep[{\varX[1]}]{\pmvA}{10:\tokT} 
  \xrightarrow{\splitOp} 
  \confDep[{\varX[2]}]{\pmvA}{8:\tokT} 
  \mid
  \confDep[{\varX[3]}]{\pmvB}{2:\tokT} 
\end{equation}

A user can transfer the ownership of any of her deposits to another user.
For instance, $\pmvA$ can $give$ her deposit $\varX[2]$ to $\pmvB$:
\begin{equation}
  \label{eq:overview:give}
  \confDep[{\varX[2]}]{\pmvA}{8:\tokT} 
  \xrightarrow{give} 
  \confDep[{\varX[4]}]{\pmvB}{8:\tokT} 
\end{equation}

After that, $\pmvB$ owns a total of $10$ units of $\tokT$ in two separate deposits,
one with $8$ units, and the other one with $2$ units.
This reflects the UTXO nature of Bitcoin:
by contrast, in account-based blockchains like Ethereum,
$\pmvB$ would have a single account storing 10 units of $\tokT$.
Now, $\pmvB$ can $\joinOp$ his two deposits, obtaining a single deposit with $10$ units of $\tokT$.
When performing the $\joinOp$ action, $\pmvB$ can also choose the owner of the new deposit,
in this case transferring it back to $\pmvA$:
\begin{equation}
  \label{eq:overview:join}
  \confDep[{\varX[4]}]{\pmvB}{8:\tokT} 
  \mid
  \confDep[{\varX[3]}]{\pmvB}{2:\tokT} 
  \xrightarrow{\joinOp} 
  \confDep[{\varX[5]}]{\pmvA}{10:\tokT} 
\end{equation}

A crucial property of the $\joinOp$ operation is that only deposits of the same 
token can be joined together. 
Thus, two deposits of $\tokT$ and $\tokTi$ with $\tokT \neq \tokTi$ cannot be joined:
\[
\confDep[{\varX[4]}]{\pmvB}{8:\tokT} 
\mid
\confDep[{\varX[6]}]{\pmvA}{2:\tokTi} 
\; \not\xrightarrow{\joinOp} 
\]

In this configuration, if both $\pmvA$ and $\pmvB$ agree, 
they can \emph{exchange} the ownership of their tokens: 
\[
\confDep[{\varX[4]}]{\pmvB}{8:\tokT} 
\mid
\confDep[{\varX[6]}]{\pmvA}{2:\tokTi} 
\; \xrightarrow{\exchangeOp} \;
\confDep[{\varX[7]}]{\pmvA}{8:\tokT} 
\mid
\confDep[{\varX[8]}]{\pmvB}{2:\tokTi} 
\]

The $\exchangeOp$ operation also supports the exchange between bitcoins and other tokens,
representing the trade of tokens. 
For instance, $\pmvA$ can buy $2$ units of $\tokTi$ from $\pmvB$ for $1 \BTC$:
\[
\confDep[{\varX[8]}]{\pmvB}{2:\tokTi} 
\mid
\confDep[{\varX[9]}]{\pmvA}{1:\BTC} 
\xrightarrow{\exchangeOp} 
\confDep[{\varX[10]}]{\pmvA}{2:\tokTi} 
\mid
\confDep[{\varX[11]}]{\pmvB}{1:\BTC} 
\]

Finally, a user can $\burnOp$ any of her deposits.
This is rendered as a state transition where the burnt deposits 
are no longer present in the target configuration.
For instance, starting from the configuration above,
$\pmvA$ can burn her $2$ units of $\tokTi$ in $\varX[10]$:
\[
\confDep[{\varX[10]}]{\pmvA}{2:\tokTi} 
\mid
\confDep[{\varX[11]}]{\pmvB}{1:\BTC} 
\xrightarrow{\burnOp}
\confDep[{\varX[11]}]{\pmvB}{1:\BTC} 
\]

\paragraph*{Bitcoin}

Although Bitcoin does not support user-defined tokens, 
it implements all the operations discussed above on its native crypto-currency.
Intuitively, each deposit corresponds to a transaction output,
and performing actions corresponds to appending a suitable transaction that redeems it.

For instance, minting bitcoins is obtained through coinbase transactions,
which are used in Bitcoin to pay rewards to miners.
We represent a coinbase transaction as follows:
\begin{nscenter}
  \small
  \begin{tabular}[t]{|l|}
    \hline
    \multicolumn{1}{|c|}{$\txT[{\ref{eq:overview:gen}}]$} \\
    \hline
    \txIn[1]{$\bot$} \\[1pt]
    \txWit[1]{$\bot$} \\[1pt]
    \txOut[1]{$\{ \txscript: \versig{\constPK[\pmvA]}{\rtx.\txWit[]{}},\, \txval: 10 \BTC \}$} \\[1pt]
    \hline
  \end{tabular}
\end{nscenter}

In general, the $\txIn{}$ field points to a previous transaction on the blockchain,
that the current one is trying to \emph{spend}.
Here, the ``undefined'' value $\bot$ characterizes $\txT[{\ref{eq:overview:gen}}]$ as a coinbase, 
since it mints bitcoins without spending any transaction.
The $\txOut{}$ field is a record, 
where $\txscript$ is a \emph{script},
and $\txval$ is the amount of bitcoins that will be redeemed 
by a subsequent transaction which points to $\txT[{\ref{eq:overview:gen}}]$ and satisfies its script.
Here, the script $\versig{\constPK[\pmvA]}{\rtx.\txWit[]{}}$ 
verifies a signature on the redeeming transaction ($\rtx$, excluding its $\txWit{}$ field) 
against $\pmvA$'s public key $\constPK[\pmvA]$.
This signature is retrieved from the $\txWit{}$ field of $\rtx$.
Since $\pmvA$ is the only user who can redeem $\txT[{\ref{eq:overview:gen}}]$, we can say that
$\txT[{\ref{eq:overview:gen}}]$ is the \emph{computational counterpart} of the deposit $\confDep[{\varX[1]}]{\pmvA}{10:\BTC}$.

To perform the $\splitOp$ action~\eqref{eq:overview:split} on $\tokT = \BTC$,
we can spend $\txT[{\ref{eq:overview:gen}}]$ with a transaction $\txT[{\ref{eq:overview:split}}]$ 
with \emph{two} outputs:
\begin{nscenter}
  \small
  \begin{tabular}[t]{|l|}
    \hline
    \multicolumn{1}{|c|}{$\txT[{\ref{eq:overview:split}}]$} \\
    \hline
    \txIn[1]{$(\txT[{\ref{eq:overview:gen}}],1)$} \\[1pt]
    \txWit[1]{$\sig{\constSK[\pmvA]}{\txT[1]}$} \\[1pt]
    \txOut[1]{$\{ \txscript: \versig{\constPK[\pmvA]}{\rtx.\txWit{}},\, \txval: 8 \BTC \}$} \\[1pt]
    \txOut[2]{$\{ \txscript: \versig{\constPK[\pmvB]}{\rtx.\txWit{}},\, \txval: 2 \BTC \}$} \\[1pt]
    \hline
  \end{tabular}
\end{nscenter}

The first output, that we denote by $(\txT[{\ref{eq:overview:split}}],1)$, corresponds to the deposit
$\confDep[{\varX[2]}]{\pmvA}{8:\BTC}$ in~\eqref{eq:overview:split}.
Instead, the output $(\txT[{\ref{eq:overview:split}}],2)$ corresponds to $\confDep[{\varX[3]}]{\pmvB}{2:\BTC}$.
These outputs can be spent independently.
For instance, performing the $\giveOp$ action in~\eqref{eq:overview:give}
corresponds to appending a transaction which spends $(\txT[{\ref{eq:overview:split}}],1)$:
\begin{nscenter}
  \small
  \begin{tabular}[t]{|l|}
    \hline
    \multicolumn{1}{|c|}{$\txT[{\ref{eq:overview:give}}]$} \\
    \hline
    \txIn[1]{$(\txT[{\ref{eq:overview:split}}],1)$} \\[1pt]
    \txWit[1]{$\sig{\constSK[\pmvA]}{\txT[{\ref{eq:overview:give}}]}$} \\[1pt]
    \txOut[1]{$\{ \txscript: \versig{\constPK[\pmvB]}{\rtx.\txWit{}},\, \txval: 8 \BTC \}$} \\[1pt]
    \hline
  \end{tabular}
\end{nscenter}

At this point, we have two unspent outputs on the blockchain: 
$(\txT[{\ref{eq:overview:split}}],2)$ and $(\txT[{\ref{eq:overview:give}}],1)$.
We can perform the $\joinOp$ action in~\eqref{eq:overview:join} by spending both of them
\emph{simultaneously} with the following transaction, which has \emph{two} inputs:
\begin{nscenter}
  \small
  \begin{tabular}[t]{|l|}
    \hline
    \multicolumn{1}{|c|}{$\txT[{\ref{eq:overview:join}}]$} \\
    \hline
    \txIn[1]{$(\txT[{\ref{eq:overview:split}}],2)$} 
    \hspace{40pt}
    \txIn[2]{$(\txT[{\ref{eq:overview:give}}],1)$} \\[1pt]
    \txWit[1]{$\sig{\constSK[\pmvB]}{\txT[{\ref{eq:overview:join}}]}$} 
    \hspace{20pt}
    \txWit[2]{$\sig{\constSK[\pmvB]}{\txT[{\ref{eq:overview:join}}]}$} \\[1pt]
    \txOut{$\{ \txscript: \versig{\constPK[\pmvB]}{\rtx.\txWit{}},\, \txval: 10 \BTC \}$} \\[1pt]
    \hline
  \end{tabular}
\end{nscenter}

\paragraph*{Implementing Bitcoin tokens  with covenants}

Although the Bitcoin script language is a bit more flexible than shown above,
it does not allow to implement \emph{on-chain} tokens.
One of the first techniques to embed on-chain tokens in 
an \emph{extended} version of Bitcoin was described in~\cite{Moser16bw}.
The technique relies on \emph{covenants}, an extension of Bitcoin scripts
which allows transactions to constrain the scripts of the redeeming ones.
For instance, let $\expe$ be an arbitrary script.
A transaction output containing the script:
\[
\expe \andE \verrec{\rtxo{}{n}}
\]
can only be redeemed by a transaction which makes $\expe$ evaluate to true,
and whose script in the $n$-th output is syntactically equal to $\expe \andE \verrec{\rtxo{}{n}}$.

Using covenants, we can mint a token by appending the transaction $\txT$ below,
where the extra field $\txarg$ is syntactic sugar for a \emph{sequence} 
of values accessible by the script
(we denote with $\seqat{\txarg}{i}$ the $i$-th element of this sequence):

\resizebox{0.95\columnwidth}{!}{
  \hspace{-10pt}
  \begin{nscenter}
    \small
    \begin{tabular}[t]{|l|}
      \hline
      \\[-9pt]
      \multicolumn{1}{|c|}{$\txT$} \\[-1pt]
      \hline
      $\cdots$ \\
      \txOut[1]{$\{ \txarg: \constPK[\pmvA], $} \\[1pt]
      \hspace{30pt} $\txscript: \versig{\seqat{\ctxo{\txarg}}{1}}{\rtx.\txWit{}} \andE$ 
      \hspace{0pt} \textcolor{Gray}{// verify signature} \\[1pt]
      \hspace{48pt} $\rtxo{\txval}{1} = 1 \andE$ 
      \hspace{39pt} \textcolor{Gray}{// preserve value} \\[1pt]
      \hspace{48pt} $\verrec{\rtxo{}{1}},$ 
      \hspace{57pt} \textcolor{Gray}{// preserve script} \\[1pt]
      \hspace{30pt} $\txval: 1 \BTC \}$ \\[1pt]
      \hline
    \end{tabular}
  \end{nscenter}
}

\bigskip
The $\txarg$ field identifies $\pmvA$ as the owner of the token:
to transfer the ownership to $\pmvB$, $\pmvA$ must spend $\txT$ with a transaction $\txTi$, 
setting its $\txarg$ to $\pmvB$'s public key.
For this to be possible, $\txTi$ must satisfy the conditions specified in $\txT$'s script:
\begin{inlinelist}
\item the $\txWit{}$ field must contain the signature of the current owner;
\item the output at index 1 must have $1 \BTC$ value, to preserve the value of the token;
\item the script at index 1 in $\txTi$ must be equal to that in $\txT$.
\end{inlinelist}
Once $\txTi$ is on the blockchain, $\pmvB$ can transfer the token to another user, 
by appending a transaction which redeems $\txTi$.

Note that the transaction $\txT$ above actually mints a \emph{non}-fungible token, 
which can be transferred from one user to another, but whose value cannot be $\splitOp$
(further, the token has a subtle flaw related to $\joinOp$ actions: we will say more on this).
The first step to turn the token into a fungible one is to support the $\splitOp$ action.
We can achieve this by adding a second element to the $\txarg$ sequence, to represent 
the number of token units deposited in the transaction output.
We can implement a splittable token as follows:

\resizebox{0.95\columnwidth}{!}{
  \hspace{-10pt}
  \begin{nscenter}
    \small
    \begin{tabular}[t]{|l|}
      \hline
      \\[-9pt]
      \multicolumn{1}{|c|}{$\txT[\!\splitOp]$} \\[-1pt]
      \hline
      $\cdots$ \\
      \txOut[1]{$\{ \txarg: \constPK[\pmvA] \, \valV, $} \\[1pt]
      \hspace{30pt} $\txscript: \versig{\seqat{\ctxo{\txarg}}{1}}{\rtx.\txWit{}} \andE$ \\[1pt]
      \hspace{48pt} $\seqat{\rtxo{\txarg}{1}}{2} + \seqat{\rtxo{\txarg}{2}}{2} = \seqat{\ctxo{\txarg}}{2} \andE$ \\[1pt] 
      \hspace{48pt} $\verrec{\rtxo{}{1}} \andE \verrec{\rtxo{}{2}} \andE$ \\[1pt]
      \hspace{48pt} $\outlen{\rtx} = 2$ \\[1pt]
      \hspace{30pt} $\txval: \cdots \}$ \\[1pt]
      \hline
    \end{tabular}
  \end{nscenter}
}

\medskip
The last two lines of the script ensure that
any transaction which redeems $\txT[\!\splitOp]$ has exactly two outputs, 
each one with the same script of $\txT[\!\splitOp]$.
The second line ensures that the $\splitOp$ preserves the number of token units
(here, $\txval$ is immaterial).

Now, let $\expe[\splitOp]$ be the script used in $\txT[\!\splitOp]$.
To extend the token with the $\joinOp$ action, first we need to add a third element to the 
$\txarg$ sequence, to encode the action performed by a transaction
(say, $G$ for $\genOp$, $S$ for $\splitOp$, and $J$ for $\joinOp$).
The extended script could have the following form:
\[
\expe \; \eqdef \;
\ifE{\seqat{\rtxo{\txarg}{1}}{3 = S}}{\expe[\splitOp]}{\expe[\joinOp]}
\]
where $\expe[\joinOp]$ implements the join functionality, \ie:
\begin{inlinelist}
\item verify the signature on the redeeming transaction; 
\item check that the redeeming transaction has exactly \emph{two} inputs and one output;
\item \label{item:overview:join1:tkval} 
  ensure that the token units are preserved;
\item \label{item:overview:join1:verrec} 
  ensure that the joined transactions represent units of the \emph{same} token.
\end{inlinelist}

For instance, consider the transactions in~\Cref{fig:overview:join1},
where $\txT[4]$ and $\txT[3]$ represent, respectively, 
the deposits $\confDep[{\varX[4]}]{\pmvB}{8:\tokT}$ and $\confDep[{\varX[3]}]{\pmvB}{2:\tokT}$.
To perform the $\joinOp$ action in~\eqref{eq:overview:join}, we must spend
$\txT[4]$ and $\txT[3]$ with the transaction $\txT[5]$:
this requires to satisfy the script $\expe$ in $\txT[4]$ and $\txT[3]$.
For condition \ref{item:overview:join1:tkval},
the script must ensure that the $10$ token units redeemed by $\txT[5]$ 
are the sum of the $8$ units in $\txT[4]$ and the $2$ units in $\txT[3]$.
For condition \ref{item:overview:join1:verrec}, the script in $\txT[4]$
should check that it is the same as that in $\txT[3]$, and \emph{viceversa}.
Hence, to implement conditions \ref{item:overview:join1:tkval}-\ref{item:overview:join1:verrec},
the script in a transaction output 
must be able to access the fields in its \emph{sibling}, \ie the transaction output which is redeemed together
(\eg, $(\txT[3],1)$ is the sibling of $(\txT[4],1)$ when appending $\txT[5]$).
However, neither Bitcoin nor its extensions with covenants \cite{BLZ20isola,Moser16bw,Oconnor17bw,Swambo20bitcoin}
allow scripts to access the siblings.

\newcommand{\txFigJoinA}{%
  \begin{tabular}[t]{|l|}
    \hline
    \multicolumn{1}{|c|}{$\txT[4]$} \\
    \hline
    $\cdots$ \\
    \txOut[1]{$\!\!\{ \txarg: \!\constPK[\pmvB] \, 8 \, S,\; \txscript: \expe,  \cdots \}\!\!\!$} \\[1pt]
    \txOut[2]{$\cdots$} \\
    \hline
  \end{tabular}
}

\newcommand{\txFigJoinB}{%
  \begin{tabular}[t]{|l|}
    \hline
    \multicolumn{1}{|c|}{$\txT[3]$} \\
    \hline
    $\cdots$ \\
    \txOut[1]{$\!\!\{ \txarg: \constPK[\pmvB] \, 2 \, S,\; \txscript: \expe, \cdots\}\!\!\!$} \\[1pt]
    \txOut[2]{$\cdots$} \\
    \hline
  \end{tabular}
}

\newcommand{\txFigJoinC}{%
  \begin{tabular}[t]{|l|}
    \hline
    \multicolumn{1}{|c|}{$\txT[5]$} \\
    \hline
    \txIn[1]{$(\txT[4],1)$} \hspace{60pt}
    \txIn[2]{$(\txT[3],1)$} \\[1pt]
    \txOut[1]{$\{ \txarg: \constPK[\pmvA] \, 10 \, J, \; \txscript: \expe,\; \txval: \cdots \}$} \\
    \hline
  \end{tabular}
}

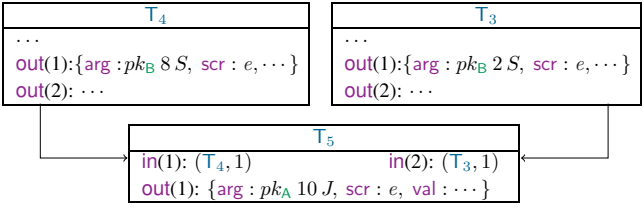
\begin{figure}[t]
  \centering
  \resizebox{\columnwidth}{!}{
    \begin{tikzpicture} 
      \node at (0,2.1) {\txFigJoinA};
      \node at (5.7,2.1) {\txFigJoinB};
      \node at (2.9,0.2) {\txFigJoinC};
      \draw [<-] (-0.5,0.3) -- (-2,0.3) -- (-2,1.24);
      \draw [<-] (6.3,0.3) -- (7.8,0.3) -- (7.8,1.24);
    \end{tikzpicture}
  } 
  \caption{A transaction $\txT[5]$ attempting to join $\txT[4]$ and $\txT[3]$.}
  \label{fig:overview:join1}
\end{figure}

\paragraph*{An insecure implementation of join}

To implement the $\joinOp$ action,
we start by extending Bitcoin scripts with an operator to access the sibling transaction outputs:
\[
\stxo{}{n}
\; \eqdef \;
\text{output redeemed by the $n$-th input of $\rtx$}
\]
Using this new operator, we can encode the conditions
\ref{item:overview:join1:tkval} and \ref{item:overview:join1:verrec}
in $\expe[\joinOp]$ as follows:
\begin{align*}
  & \seqat{\rtxo{\txarg}{1}}{2} = \seqat{\stxo{\txarg}{1}}{2} + \seqat{\stxo{\txarg}{2}}{2}
  && \text{\ref{item:overview:join1:tkval}}
  \\
  & \verrec{\stxo{}{1}} \andE \verrec{\stxo{}{2}}
  && \text{\ref{item:overview:join1:verrec}}
\end{align*}

Although this implementation of $\expe[\joinOp]$ correctly encodes the conditions,
it introduces a security vulnerability:
an adversary can join two deposits of \emph{different} tokens.
The attack is exemplified in~\Cref{fig:overview:join2}.
The transactions $\txT[\pmvA]$ and $\txT[\pmv{M}]$ mint $10$ units of \emph{different} tokens,
and transaction $\txT[3]$ joins them into a single deposit of the \emph{same} token.
Ideally, to counter this attack, $\expe[\joinOp]$ should check not only the sibling,
but also its ancestors until the minting transaction, and verify that it corresponds to
the minting ancestor of the current transaction output.
Although this would be possible by adding script operators that can go up the transaction graph
at an arbitrary depth, this would be highly inefficient from the point of view of miners,
who should record the \emph{whole} transaction graph, instead of just the set of unspent transactions (UTXO).

\newcommand{\txFigJoinAttackGenA}{%
  \begin{tabular}[t]{|l|}
    \hline
    \multicolumn{1}{|c|}{$\txT[A]$} \\
    \hline
    $\cdots$ \\
    \txOut[1]{$\{ \txarg: \!\constPK[\pmvA] \, 10 \, G, $} \\[1pt]
    \hspace{34pt} $\txscript: \expe,\, \txval: \cdots \}\;$ \\
    \hline
  \end{tabular}
}

\newcommand{\txFigJoinAttackGenM}{%
  \begin{tabular}[t]{|l|}
    \hline
    \multicolumn{1}{|c|}{$\txT[M]$} \\
    \hline
    $\cdots$ \\
    \txOut[1]{$\{ \txarg: \!\constPK[\pmv{M}] \, 10 \, G, $} \\[1pt]
    \hspace{34pt} $\txscript: \expe,\, \txval: \cdots \}\;$ \\
    \hline
  \end{tabular}
}

\newcommand{\txFigJoinAttackSplitA}{%
  \begin{tabular}[t]{|l|}
    \hline
    \multicolumn{1}{|c|}{$\txT[1]$} \\
    \hline
    \txIn[1]{$(\txT[\pmvA],1)$} \\[1pt]
    \txOut[1]{$\{ \txarg: \!\constPK[\pmvA] \, 8 \, S, $} \\[1pt]
    \hspace{34pt} $\txscript: \expe,\, \txval: \cdots \}$ \\[1pt]
    \txOut[2]{$\cdots$} \\
    \hline
  \end{tabular}
}

\newcommand{\txFigJoinAttackSplitM}{%
  \begin{tabular}[t]{|l|}
    \hline
    \multicolumn{1}{|c|}{$\txT[2]$} \\
    \hline
    \txIn[1]{$(\txT[\pmv{M}],1)$} \\[1pt]
    \txOut[1]{$\{ \txarg: \!\constPK[\pmvA] \, 7 \, S, $} \\[1pt]
    \hspace{34pt} $\txscript: \expe,\, \txval: \cdots \}$ \\[1pt]
    \txOut[2]{$\cdots$} \\
    \hline
  \end{tabular}
}

\newcommand{\txFigJoinAttackJoinM}{%
  \begin{tabular}[t]{|l|}
    \hline
    \multicolumn{1}{|c|}{$\txT[3]$} \\
    \hline
    \txIn[1]{$(\txT[1],1)$} \hspace{60pt}
    \txIn[2]{$(\txT[2],1)$} \\[1pt]
    \txOut[1]{$\{ \txarg: \!\constPK[\pmvA] \, 15 \, J,\; \txscript: \expe,\, \txval: \cdots \}$} \\
    \hline
  \end{tabular}
}

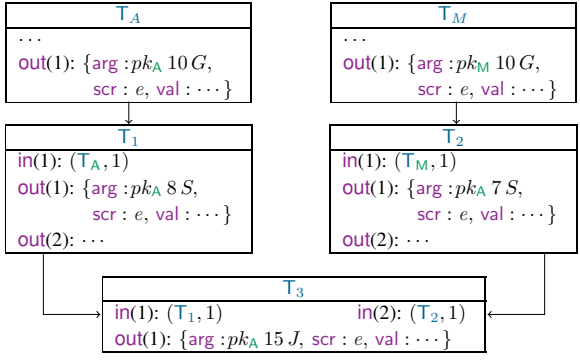
\begin{figure}[t]
  \centering
  \resizebox{0.9\columnwidth}{!}{
    \begin{tikzpicture} 
      \node at (0,4.5) {\txFigJoinAttackGenA};
      \node at (5.7,4.5) {\txFigJoinAttackGenM};
      \draw [->] (0,3.64) -- (0,3.25);
      \draw [->] (5.75,3.64) -- (5.75,3.25);
      \node at (0,2.1) {\txFigJoinAttackSplitA};
      \node at (5.7,2.1) {\txFigJoinAttackSplitM};
      \node at (2.9,-0.1) {\txFigJoinAttackJoinM};
      \draw [<-] (-0.5,-0.1) -- (-1.5,-0.1) -- (-1.5,0.97);
      \draw [<-] (6.3,-0.1) -- (7.3,-0.1) -- (7.3,0.97);
    \end{tikzpicture}
  } 
  \caption{A $\joinOp$ attack merging two different tokens.}
  \label{fig:overview:join2}
\end{figure}

\paragraph*{Neighbourhood covenants}

To address this issue, we use an operator which can go up the transaction graph only one level,
\ie up to the \emph{parent} of the current transaction.
Hence, implementing our Bitcoin extension with \emph{neighbourhood covenants}
requires miners to just record the UTXOs and their parents. 
By exploiting this new covenant, we can thwart the $\joinOp$ attack of~\Cref{fig:overview:join2},
and eventually obtain a secure and efficient implementation of fungible tokens.
We now sketch the script $\tokScript$ which implements tokens.
First, we add a fourth element to the $\txarg$ sequence, to record in each transaction output
the \emph{identifier} of the token deposited in that output.
As an identifier, we use the hash of the \emph{parent} of the minting transaction,
which we access through the script $\txid{\ptxo{}{1}}$.
When the script in a minting transaction 
(\eg, $\txT[\pmvA]$ and $\txT[\pmv{M}]$ in \Cref{fig:overview:join2})
is evaluated, 
it ensures that the $\txarg$ field of the minting transaction
actually contains the token identifier:
\[
\expe[\genOp] \eqdef \cdots \andE \seqat{\ctxo{\txarg}}{4} = \txid{\ptxo{}{1}}
\]
Then, the sub-scripts corresponding to all other token actions
check that the redeeming transaction preserves the token identifier.
For instance, in the $\joinOp$ sub-script,
besides checking conditions \ref{item:overview:join1:tkval} and \ref{item:overview:join1:verrec}
as shown before, we add the condition:
\[
\expe[\joinOp] \eqdef \cdots \andE \seqat{\ctxo{\txarg}}{4} = \seqat{\txof{\rtxo{}{1}}{\txarg}}{4}
\]
In~\Cref{fig:overview:join3} we show how this resolves the attack of~\Cref{fig:overview:join2}.
In order to append the malicious $\joinOp$ transaction $\txT[3]$,
we must satisfy the script $\tokScript$ in both $\txT[1]$ and $\txT[2]$.
These scripts check that the $tokid$ in the redeeming transaction $\txT[3]$
is equal to the identifiers of the two branches, $\hashSem{\txTi[\pmvA],1}$ and $\hashSem{\txTi[\pmv{M}],1}$:
by collision resistance of the hash function, this is not possible.

\newcommand{\txFigJoinOkGenA}{%
  \begin{tabular}[t]{|l|}
    \hline
    \multicolumn{1}{|c|}{$\txT[A]$} \\
    \hline
    \txIn[1]{$(\txTi[\pmvA],1)$} \\
    \txOut[1]{$\{ \txarg: \!\constPK[\pmvA] \, 10 \, G\, \hashSem{\txTi[\pmvA],1},\!\!\!$} \\[1pt]
    \hspace{34pt} $\txscript: \tokScript,\, \txval: \cdots \}\;$ \\
    \hline
  \end{tabular}
}

\newcommand{\txFigJoinOkGenM}{%
  \begin{tabular}[t]{|l|}
    \hline
    \multicolumn{1}{|c|}{$\txT[M]$} \\
    \hline
    \txIn[1]{$(\txTi[\pmv{M}],1)$} \\
    \txOut[1]{$\{ \txarg: \!\constPK[\pmv{M}] \, 10 \, G\, \hashSem{\txTi[\pmv{M}],1},\!$} \\[1pt]
    \hspace{34pt} $\txscript: \tokScript,\, \txval: \cdots \}\;$ \\
    \hline
  \end{tabular}
}

\newcommand{\txFigJoinOkSplitA}{%
  \begin{tabular}[t]{|l|}
    \hline
    \multicolumn{1}{|c|}{$\txT[1]$} \\
    \hline
    \txIn[1]{$(\txT[\pmvA],1)$} \\[1pt]
    \txOut[1]{$\{ \txarg: \constPK[\pmvA] \, 8 \, S \, \hashSem{\txTi[\pmvA],1},$} \\[1pt]
    \hspace{34pt} $\txscript: \tokScript,\, \txval: \cdots \}$ \\[1pt]
    \txOut[2]{$\cdots$} \\
    \hline
  \end{tabular}
}

\newcommand{\txFigJoinOkSplitM}{%
  \begin{tabular}[t]{|l|}
    \hline
    \multicolumn{1}{|c|}{$\txT[2]$} \\
    \hline
    \txIn[1]{$(\txT[\pmv{M}],1)$} \\[1pt]
    \txOut[1]{$\{ \txarg: \constPK[\pmvA] \, 7 \, S \, \hashSem{\txTi[\pmv{M}],1},\;\;\,$} \\[1pt]
    \hspace{34pt} $\txscript: \tokScript,\, \txval: \cdots \}$ \\[1pt]
    \txOut[2]{$\cdots$} \\
    \hline
  \end{tabular}
}

\newcommand{\txFigJoinOkJoinM}{%
  \begin{tabular}[t]{|l|}
    \hline
    \multicolumn{1}{|c|}{$\txT[3]$} \\
    \hline
    \txIn[1]{$(\txT[1],1)$} \hspace{100pt}
    \txIn[2]{$(\txT[2],1)$} \\[1pt]
    \txOut[1]{$\{ \txarg: \!\constPK[\pmvA] \, 15 \, J \, tokid,\; \txscript: \tokScript,\; \txval: \cdots \}$} \\
    \hline
  \end{tabular}
}

\begin{figure}[t]
  \centering
  \resizebox{\columnwidth}{!}{
    \begin{tikzpicture} 
      \node at (0,4.5) {\txFigJoinOkGenA};
      \node at (5.7,4.5) {\txFigJoinOkGenM};
      \draw [->] (0,3.64) -- (0,3.25);
      \draw [->] (5.75,3.64) -- (5.75,3.25);
      \node at (0,2.1) {\txFigJoinOkSplitA};
      \node at (5.7,2.1) {\txFigJoinOkSplitM};
      \node at (2.9,-0.1) {\txFigJoinOkJoinM};
      \draw [<-] (-1.17,-0.1) -- (-2,-0.1) -- (-2,0.97);
      \draw [<-] (7,-0.1) -- (7.8,-0.1) -- (7.8,0.97);
      \draw [line width=0.5mm, red] (7.6,0.6) -- (8.0,0.2);
      \draw [line width=0.5mm, red] (7.6,0.2) -- (8.0,0.6);
    \end{tikzpicture}
  } 
  \caption{Thwarting a $\joinOp$ attack.}
  \label{fig:overview:join3}
\end{figure}
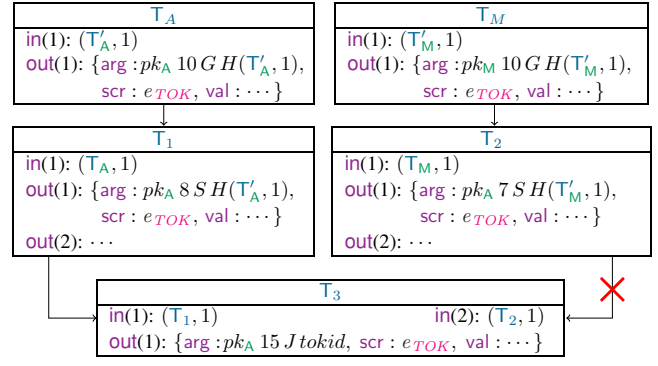

\paragraph*{Thwarting forgery attacks}

Note that minting a token amounts to appending a $\genOp$ transaction.
For instance, in~\Cref{fig:overview:gen-attack},
$\pmvA$ appends the $\genOp$ transaction $\txT[\pmvA]$,
which mints 10 units of a fresh token
with identifier $\hashSem{\txTi[\pmvA],1}$.
To validate this operation, the script $\btcScript$ in
$\txTi[\pmvA]$ is executed, 
verifying only that $\txT[\pmvA]$ is signed by $\pmvA$.
Consequently, when minting a token, $\pmvA$ can choose arbitrary 
values for the other fields of the transaction $\txT[\pmvA]$:
in particular, $\pmvA$ could choose a token identifier
(\ie, the value in $\txarg.4$) which is different from the correct one.

For instance, \Cref{fig:overview:gen-attack} shows
the adversary $\pmv{M}$ \emph{forging} 10 additional units of the token
with identifier $\hashSem{\txTi[\pmvA],1}$.
This is achieved by appending the transaction $\txT[\pmv{M}]$, 
which spends a standard transaction $\txTi[\pmv{M}]$,
and setting its $\txarg.4$ field to $\hashSem{\txTi[\pmvA],1}$.
However, these forged token units are \emph{unspendable}:
when $\pmv{M}$ attempts to spend these units
by appending $\txT[2]$, the $\tokScript$ script in $\txT[\pmv{M}]$
is finally executed.
There, the check $\seqat{\ctxo{\txarg}}{4} = \txid{\ptxo{}{1}}$
in $\expe[\genOp]$ fails,
because $\seqat{\ctxo{\txarg}}{4}$ evaluates to $\hashSem{\txTi[\pmvA],1}$,
while $\txid{\ptxo{}{1}}$ evaluates to $\hashSem{\txTi[\pmv{M}],1}$.
In general, although forging tokens is unavoidable,
the script $\tokScript$ ensures that forged tokens are unspendable,
so making forgery immaterial.


\newcommand{\txFigForgeOkDepA}{%
  \begin{tabular}[t]{|l|}
    \hline
    \multicolumn{1}{|c|}{$\txTi[A]$} \\
    \hline
    \txIn[1]{$\cdots$} \\
    \txOut[1]{$\{ \txarg: \!\constPK[\pmvA]$} \\[1pt]
    \hspace{34pt} $\txscript: \versig{\seqat{\ctxo{\txarg}}{1}}{\rtx.\txWit{}} \!\!\!$ \\
    \hspace{34pt} $\txval: \cdots \}\;$ \\
    \hline
  \end{tabular}
}

\newcommand{\txFigForgeOkDepM}{%
  \begin{tabular}[t]{|l|}
    \hline
    \multicolumn{1}{|c|}{$\txTi[\pmv{M}]$} \\
    \hline
    \txIn[1]{$\cdots$} \\
    \txOut[1]{$\{ \txarg: \!\constPK[\pmv{M}]$} \\[1pt]
    \hspace{34pt} $\txscript: \versig{\seqat{\ctxo{\txarg}}{1}}{\rtx.\txWit{}} \!\!\!$ \\
    \hspace{34pt} $\txval: \cdots \}\;$ \\
    \hline
  \end{tabular}
}

\newcommand{\txFigForgeOkGenA}{%
  \begin{tabular}[t]{|l|}
    \hline
    \multicolumn{1}{|c|}{$\txT[\pmvA]$} \\
    \hline
    \txIn[1]{$(\txTi[\pmvA],1)$} \\
    \txOut[1]{$\{ \txarg: \!\constPK[\pmvA] \, 10 \, G\, \hashSem{\txTi[\pmvA],1},\!\!\!\!$} \\[1pt]
    \hspace{34pt} $\txscript: \tokScript,\, \txval: \cdots \}\;$ \\
    \hline
  \end{tabular}
}

\newcommand{\txFigForgeOkGenM}{%
  \begin{tabular}[t]{|l|}
    \hline
    \multicolumn{1}{|c|}{$\txT[M]$} \\
    \hline
    \txIn[1]{$(\txTi[\pmv{M}],1)$} \\
    \txOut[1]{$\{ \txarg: \!\constPK[\pmv{M}] \, 10 \, G\, \hashSem{\txTi[\pmv{A}],1},\!\!\!\!$} \\[1pt]
    \hspace{34pt} $\txscript: \tokScript,\, \txval: \cdots \}\;$ \\
    \hline
  \end{tabular}
}

\newcommand{\txFigForgeOkSplitA}{%
  \begin{tabular}[t]{|l|}
    \hline
    \multicolumn{1}{|c|}{$\txT[1]$} \\
    \hline
    \txIn[1]{$(\txT[\pmvA],1)$} \\[1pt]
    \txOut[1]{$\{ \txarg: \constPK[\pmvA] \, 8 \, S \, \hashSem{\txTi[\pmvA],1},$} \\[1pt]
    \hspace{34pt} $\txscript: \tokScript,\, \txval: \cdots \}$ \\[1pt]
    \txOut[2]{$\cdots$} \\
    \hline
  \end{tabular}
}

\newcommand{\txFigForgeOkSplitM}{%
  \begin{tabular}[t]{|l|}
    \hline
    \multicolumn{1}{|c|}{$\txT[2]$} \\
    \hline
    \txIn[1]{$(\txT[\pmv{M}],1)$} \\[1pt]
    \txOut[1]{$\{ \txarg: \constPK[\pmv{M}] \, 7 \, S \, \hashSem{\txTi[\pmv{A}],1},\;\;\,$} \\[1pt]
    \hspace{34pt} $\txscript: \tokScript,\, \txval: \cdots \}$ \\[1pt]
    \txOut[2]{$\cdots$} \\
    \hline
  \end{tabular}
}

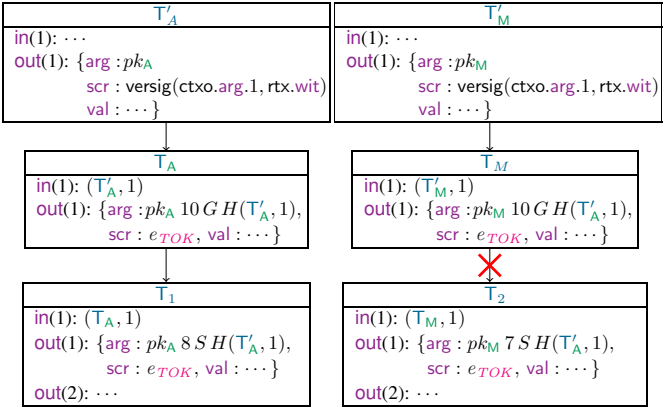
\begin{figure}[t]
  \resizebox{1.03\columnwidth}{!}{
    \begin{tikzpicture} 
      \node at (0,7) {\txFigForgeOkDepA};
      \node at (6.05,7) {\txFigForgeOkDepM};
      \draw [->] (0,5.93) -- (0,5.39);
      \draw [->] (5.9,5.91) -- (5.9,5.39);
      \node at (0,4.5) {\txFigForgeOkGenA};
      \node at (6,4.5) {\txFigForgeOkGenM};
      \draw [->] (0,3.63) -- (0,3);
      \draw [->] (5.9,3.63) -- (5.9,3);
      \node at (0,1.85) {\txFigForgeOkSplitA};
      \node at (6,1.85) {\txFigForgeOkSplitM};
      \draw [line width=0.5mm, red] (5.7,3.5) -- (6.1,3.1);
      \draw [line width=0.5mm, red] (5.7,3.1) -- (6.1,3.5);
    \end{tikzpicture}
  } 
  \caption{Thwarting a forgery attack.}
  \label{fig:overview:gen-attack}
\end{figure}

\section{A symbolic model of tokens}
\label{sec:tokens:s}

Let $\pmvA, \pmvB, \ldots$ range over \emph{users},
and let $\tokT, \tokTi, \ldots$ range over \emph{tokens},
encompassing both user-defined ones and bitcoins ($\BTC$).
A term
\(
\confDep[\varX]{\pmvA}{\valV:\tokT}
\)
represents a \emph{deposit} of $\valV \in \Nat$ units of the token $\tokT$ owned by $\pmvA$ 
(the index $\varX$ is an unique identifier of the deposit).
A term
\(
\confAuth[\varX]{\pmvA}{\labS}
\)
represents $\pmvA$'s \emph{authorization} to perform the \emph{action} $\labS$
on the deposit $\varX$.
The possible token actions are the following:
\begin{itemize}

\item $\actGen{\varX}{\valV}$ represents the act of spending a bitcoin deposit $\varX$
  to mint $\valV$ units of a new token.
  The owner of these units is the user who owned the deposit $\varX$.

\item $\actBurn{\vec{\varX}}{\varY}$ represents the act of destroying a sequence of deposits $\vec{\varX}$,
moving them to an unspendable deposit $\varY$.

\item $\actSplit{\varX}{\valV}{\pmvB}$ represents the act of splitting a deposit $\varX$
(say, containing $\valV+\valVi$ units of a token $\tokT$) in two deposits of $\tokT$.
The first one of these deposits is owned by the same owner of $\varX$, and contains $\valV$ token units.
The second one if owned by $\pmvB$, and contains the remaining $\valVi$ units.

\item $\actJoin{\varX}{\varY}{\pmvC}$ represents the joining of two deposits $\varX$ and $\varY$
of the same token, into a new deposit, owned by $\pmvC$.

\item $\actExchange{\varX}{\varY}$ represents the act of atomically exchanging
the owners of the two deposits $\varX$ and $\varY$ (not both of bitcoins).
In particular, when one of the deposits stores $\BTC$,
this action represents buying/selling tokens for bitcoins.

\item $\actGive{\varX}{\pmvB}$ represents a donation of the deposit $\varX$ to $\pmvB$.

\end{itemize}

Users follow a common pattern to perform token actions:
\begin{inlinelist}
\item first, the involved users grant their authorization on the action;
\item once all the needed authorizations have been granted, the action can actually be performed. 
\end{inlinelist}

A \emph{configuration} $\confG$ is a compositions of deposits and authorizations.
We assume that configurations form a commutative monoid under the composition operator $\mid$,
and we use $\gnil$ to denote the empty configuration.
We require that
if $\confDep[\varX]{\pmvA}{\valV:\tokT}$ and $\confDep[\varXi]{\pmvB}{\valVi:\tokTi}$ both occur in $\confG$,
then $\varX \neq \varXi$.
We define a transition semantics between configurations in~\Cref{fig:tokens-s:semantics}.
Transitions are decorated with labels, which describe the performed actions.
The rules for granting authorizations (\Cref{fig:tokens-s:auth:semantics}
are straightforward.

\begin{figure*}
  \centering
  \small
  \begin{tabular}{c}
    \(
    \irule
    {
    \confG = \confAuth[\varX]{\pmvA}{\actGen{\varX}{\valV}} \mid \confGi \quad
    \valV > 0 \quad
    y,\tokT \; \text{fresh}
    }
    {\confDep[\varX]{\pmvA}{0:\BTC} 
    \mid \confG
    \xrightarrow{\actGen{\varX}{\valV}} 
    \confDep[\varY]{\pmvA}{\valV:\tokT} 
    \mid \confGi
    }
    \genRule[small]
    \)
    \hspace{0pt}
    \(
    \irule
    { 
    \confG = \big( \mmid_{i \in 1..n} \, \confAuth[{\varX[i]}]{\pmvA[i]}{\actBurn{x_1 \cdots x_n}{y}} \big) \mid \confGi \quad
    n=1 \lor (n \geq 1 \land \forall i: \tokT[i]=\BTC)
    }
    {\big( \mmid_{i \in 1..n} \confDep[x_i]{\pmvA[i]}{\valV[i]:\tokT[i]} \big) 
    \mid \confG
    \xrightarrow{\actBurn{x_1 \cdots x_n}{y}}
    \confGi}
    \burnRule[small]
    \)
    \\[20pt]
    \(
    \irule
    {
    \confG = \confAuth[\varX]{\pmvA}{\actSplit{\varX}{\valV}{\pmvB}} \mid \confGi \quad
    \valV,\valVi \geq 0 \quad
    y,y' \; \text{fresh}
    }
    {\confDep[x]{\pmvA}{(\valV + \valVi):\tokT} \mid \confG
    \xrightarrow{\actSplit{\varX}{\valV}{\pmvB}}
    \confDep[y]{\pmvA}{\valV:\tokT} \mid \confDep[y']{\pmvB}{\valVi:\tokT}
    \mid \confGi
    }
    \splitRule[small]
    \)
    \hspace{0pt}
    \(
    \irule
    {
    \confG = \confAuth[\varX]{\pmvA}{\actJoin{\varX}{\varY}{\pmvC}} \mid \confAuth[\varY]{\pmvB}{\actJoin{\varX}{\varY}{\pmvC}} \mid \confGi \quad
    z \; \text{fresh}
    }
    {\confDep[x]{\pmvA}{\valV:\tokT} \mid \confDep[y]{\pmvB}{\valVi:\tokT} \mid \confG
    \xrightarrow{\actJoin{x}{y}{\pmvC}} 
    \confDep[z]{\pmvC}{(\valV + \valVi):\tokT} \mid \confGi
    }
    \joinRule[small]
    \)
    \\[20pt]
    \(
    \irule
    {
    \confG = \confAuth[\varX]{\pmvA}{\actExchange{x}{y}} \mid \confAuth[\varY]{\pmvB}{\actExchange{x}{y}} \mid \confGi \quad
    \tokT \neq \BTC \quad 
    x',y' \; \text{fresh} 
    }
    {\confDep[x]{\pmvA}{\valV:\tokT} \mid \confDep[y]{\pmvB}{\valVi:\tokTi} 
    \mid \confG
    \xrightarrow{\actExchange{x}{y}} 
    \confDep[x']{\pmvA}{\valVi:\tokTi} \mid \confDep[y']{\pmvB}{\valV:\tokT}
    \mid \confGi
    }
    \exchangeRule[small]
    \)
    \qquad
    \(
    \irule
    {
    \confG = \confAuth[\varX]{\pmvA}{\actGive{x}{\pmvB}} \mid \confGi \quad
    y \; \text{fresh} 
    }
    {\confDep[x]{\pmvA}{\valV:\tokT}
    \mid \confG
    \xrightarrow{\actGive{x}{\pmvB}} 
    \confDep[y]{\pmvB}{\valV:\tokT}
    \mid \confGi
    }
    \giveRule[small]
    \)
  \end{tabular}
  \caption{Semantics of token actions.}
  \label{fig:tokens-s:semantics}
\end{figure*}

Rule \genRule consumes a bitcoin deposit $\varX$ owned by $\pmvA$ 
to generate $\valV$ units of a new token $\tokT$, which are stored in a fresh deposit $\varY$.
Note that $\pmvA$'s authorization is required to perform the action.
For simplicity, we assume that minting tokens has no cost:
it would be straightforward to adapt the rule to require a minting fee.
Rule \burnRule removes from the configuration a single token deposit 
(when $n=1$ and $\tokT[1] \neq \BTC$), 
or atomically removes a sequence of bitcoin deposits (when $\tokT[i] = \BTC$ for all $i$).
Rule \splitRule divides a deposit $\varX$ in two fresh deposits $\varY$ and $\varZ$,
preserving the number of token units.
%
Rule \joinRule allows $\pmvA$ and $\pmvB$ to merge two deposits of a token $\tokT$, 
preserving the amount of token units, and transferring the new deposit to $\pmvC$.
Rule \exchangeRule allows $\pmvA$ and $\pmvB$ to swap two deposits, containing either
user-defined tokens or bitcoins.
Finally, rule \giveRule allows $\pmvA$ to donate one of her deposits to another user.

The transition relation $\xrightarrow{}$ is non-deterministic, 
because of the fresh names generated for deposits and tokens:
however, given a transition $\confG \xrightarrow{\labS} \confGi$
the label $\labS$ is uniquely determined from $\confG$ and $\confGi$.
%
\label{def:runS}
A \emph{symbolic run} $\runS$ is a (possibly infinite) sequence 
$\confG[0] \confG[1] \cdots$,
where $\confG[0]$ contains only $\BTC$ deposits, and for all $i \geq 0$
there exists some (unique) $\labS[i]$ such that $\confG[i] \xrightarrow{\labS[i]} \confG[i+1]$.
For all $i \geq 0$, we denote with $\runat{\runS}{i}$ the $i$-th element of the run,
when this element exists.
If $\runS$ is finite, we denote its length as $|\runS|$,
and we write $\confS{\runS}$ for its last configuration, \ie $\runat{\runS}{|\runS|-1}$.

\begin{defn}[Token balance]
  \label{def:token-s:tokval}
  We define the balance of a token $\tokT \neq \BTC$ 
  in a configuration $\confG$ inductively as follows:
  \[
  \begin{array}{ll}
    \tokval{\tokT}{\gnil} = 0
    & \;\tokval{\tokT}{\confG \mid \confGi} = \tokval{\tokT}{\confG} + \tokval{\tokT}{\confGi}
    \\[5pt]
    \tokval{\tokT}{\confDep[\varX]{\pmvA}{\valV:\tokT}} = \valV
    & \;\tokval{\tokT}{\confDep[\varX]{\pmvA}{\valV:\tokTi}} = 0 \quad (\tokTi \neq \tokT)
  \end{array}
  \]
\end{defn}

The following~\namecref{lem:tokens-s:tokval-preservation} establishes a basic preservation property:
the balance of a token after a run is equal to the \emph{minted} value minus the \emph{burnt} value, 
defined as:
\begin{align*}
  & \genval{\tokT}{\runS} = v
    \;\; \text{if } \exists i : 
    \begin{array}{l}
      \runat{\runS}{i} \xrightarrow{\actGen{\varX}{\valV}} \runat{\runS}{i+1},\; \text{and}
      \\
      \text{$\tokT$ occurs in $\runat{\runS}{i+1}$ but not in $\runat{\runS}{i}$}
    \end{array}
  \\
  & \burnval{\tokT}{\runS} = \sum \setcomp{\valV}{\exists i : 
    \begin{array}{l} \runat{\runS}{i} \xrightarrow{\actBurn{\varX}{\varY}} \runat{\runS}{i+1},\; \text{and} \\ 
      \runat{\runS}{i} = \confG \mid \confDep[\varX]{\pmvA}{\valV:\tokT}
    \end{array}
  }
\end{align*}

\begin{lem} 
  \label{lem:tokens-s:tokval-preservation}
  Let $\runS$ be a finite symbolic run.
  For all $\tokT \neq \BTC$:
  \[
  \tokval{\tokT}{\runS} 
  \; = \; 
  \genval{\tokT}{\runS} - \burnval{\tokT}{\runS}
  \]
\end{lem}
\begin{proof} 
  By induction on $\runS$, and by case analysis on each step.
  Inspecting each symbolic semantics rule, we can see that each step
  preserves the amount of token units, except for minting
  ($\genRule{}$) and burning ($\burnRule{}$), which are explicitly
  taken into account by the equation.
\end{proof}

  \newcommand{\ftokensplit}{{\sf S}}


\begin{figure*}[t!]
  \centering
  \small
  \begin{tabular}{c}
    \(
    \irule
    {\valV > 0}
    {\confDep[x]{\pmvA}{0:\BTC} \mid \confG
    \xrightarrow{\labAuth[x]{\pmvA}{\actGen{x}{\valV}}}
    \confDep[x]{\pmvA}{0:\BTC} \mid 
    \confAuth[x]{\pmvA}{\actGen{x}{\valV}} \mid \confG
    }
    \smallnrule{[AuthGen]}
    \)
    \\[20pt]
    \(
    \irule
    {\vec{x} = x_1 \cdots x_n \quad j \in 1..n \quad
    y \text{ fresh (except in burn auth for $\vec{x}$)} \quad
    n=1 \lor (n \geq 1 \land \forall i: \tokT[i]=\BTC)
    }
    {\big( \mmid_{i \in 1..n} \confDep[x_i]{\pmvA[i]}{\valV[i]:\tokT[i]} \big) \mid \confG
    \xrightarrow{\labAuth[x_j]{\pmvA[j]}{\actBurn{\vec{x}}{y}}}
    \big( \mmid_{i \in 1..n} \confDep[x_i]{\pmvA[i]}{\valV[i]:\tokT[i]} \big) \mid 
    \confAuth[{\varX[j]}]{\pmvA[j]}{\actBurn{\vec{x}}{y}} \mid 
    \confG
    }
    \smallnrule{[AuthBurn]}
    \)
    \\[20pt]
    \(
    \irule
    {\valV,\valVi \geq 0}
    {\confDep[x]{\pmvA}{(\valV + \valVi):\tokT} \mid \confG
    \xrightarrow{\labAuth[x]{\pmvA}{\actSplit{\varX}{\valV}{\pmvB}}} 
    \confDep[x]{\pmvA}{(\valV+\valVi):\tokT} \mid 
    \confAuth[x]{\pmvA}{\actSplit{\varX}{\valV}{\pmvB}} \mid \confG
    }
    \smallnrule{[AuthSplit]}
    \)
    \\[20pt]
    \(
    \irule
    {(\pmvP,z) \in \setenum{(\pmvA,x),(\pmvB,y)}}
    {\confDep[x]{\pmvA}{\valV:\tokT} \mid \confDep[y]{\pmvB}{\valVi:\tokT} \mid \confG
    \xrightarrow{\labAuth[z]{\pmvP}{\actJoin{x}{y}{\pmvC}}}
    \confDep[x]{\pmvA}{\valV:\tokT} \mid \confDep[y]{\pmvB}{\valVi:\tokT} \mid 
    \confAuth[z]{\pmvP}{\actJoin{x}{y}{\pmvC}} \mid \confG
    }
    \smallnrule{[AuthJoin]}
    \)
    \\[20pt]
    \(
    \irule
    {(\pmvP,z) \in \setenum{(\pmvA,x),(\pmvB,y)} \quad
    \tokT \neq \BTC}
    {\confDep[\varX]{\pmvA}{\valV:\tokT} \mid \confDep[\varY]{\pmvB}{\valVi:\tokTi} \mid \confG
    \xrightarrow{\labAuth[z]{\pmvP}{\actExchange{\varX}{\varY}}}
    \confDep[\varX]{\pmvA}{\valV:\tokT} \mid \confDep[\varY]{\pmvB}{\valVi:\tokTi} \mid
    \confAuth[z]{\pmvP}{\actExchange{\varX}{\varY}} \mid \confG
    }
    \smallnrule{[AuthExchange]}
    \)
    \\[20pt]
    \(
    \irule
    {}
    {\confDep[\varX]{\pmvA}{\valV:\tokT} \mid \confG
    \xrightarrow{\labAuth[\varX]{\pmvA}{\actGive{\varX}{\pmvB}}}
    \confDep[\varX]{\pmvA}{\valV:\tokT} \mid
    \confAuth[\varX]{\pmvA}{\actGive{\varX}{\pmvB}} \mid \confG
    }
    \smallnrule{[AuthGive]}
    \)
  \end{tabular}
  \caption{Semantics of authorizations.}
  \label{fig:tokens-s:auth:semantics}
\end{figure*}

\section{Bitcoin transactions}
\label{sec:bitcoin}

In this~\namecref{sec:bitcoin} we recall the functionality of Bitcoin.
To this purpose we rely on the formal model of~\cite{bitcointxm},
simplifying or omitting the parts that are irrelevant for our subsequent technical development
(\eg, we abstract from the fact that, in the Bitcoin blockchain, transactions are grouped into blocks).

\paragraph*{Transactions}

Following the formalization in~\cite{bitcointxm}, 
we represent transactions as records with the following fields:%
\footnote{Bitcoin transactions can also impose time constraints on when they can be appended to the blockchain, or when they can be redeemed. Since time constraints are immaterial for our technical development, we omit them.}.
\begin{itemize}
\item $\txIn{}$ is the list of \emph{inputs}. 
  Each of these inputs is a \emph{transaction output} $(\txT,i)$, 
  referring to the $i$-th output field of $\txT$.
\item $\txWit{}$ is the list of \emph{witnesses},
  of the same length as $\txIn{}$.
  Intuitively, for each $(\txT,i)$ in the $\txIn{}$ field, 
  the witness at the same index must make the $i$-th output script 
  of $\txT$ evaluate to true.
\item $\txOut{}$ is the list of \emph{outputs}.  
  Each output is a record of the form $\{ \txscript: \expe, \txval: \valV \}$, 
  where $\expe$ is a script, and $\valV \geq 0$.
\end{itemize} 

We let $\txf$ range over transaction fields, 
and we denote with $\txT.\txf$ the content of field $\txf$ of transaction $\txT$.
We write $\txT.\txf(i)$ for the $i$-th element of the sequence $\txT.\txf$, when in range;
when the sequence has exactly one element, we write $\txT.\txf$ for $\txT.\txf(1)$.
For transaction outputs $(\txT,i)$, we interchangeably use the notation
$(\txT,i)$ and $\txT.\txOut[]{}(i)$,
and we use the notation above to access their sub-fields $\txscript$ and $\txval$.
When clear from the context, we just write the name $\pmvA$ of a user in place
of her public/private keys, \eg we write $\versig{\pmvA}{\expe}$ for $\versig{\constPK[\pmvA]}{\expe}$,
and $\sig{\pmvA}{\txT}$ for $\sig{\constSK[\pmvA]}{\txT}$.

\paragraph*{Scripts}
\label{sec:bitcoin-scripts}

Bitcoin scripts are terms with the following syntax:
\begin{align*}
  \expe 
  \; \bnfdef \;
  & \valV  
  && \text{constant (integer or bitstring)}
  \\[-1pt]
  \bnfmid 
  & \expe \circ \expe 
  && \text{operators $(\circ \in \{+, -, =, <\})$}
  \\[-1pt]
  \bnfmid 
  & \ifE{\expe}{\expe}{\expe}  
  && \text{conditional}
  \\[-1pt]
  \bnfmid 
  & \seqat{\expe}{n}
  && \text{$n$-th element of sequence $\expe$ ($n \in \Nat$)}
  \\[-1pt]
  \bnfmid
  & \rtx.\txWit{}
  && \text{witnesses of the redeeming tx}
  \\[-1pt]
  \bnfmid
  & \sizeE{\expe} 
  && \text{size (number of bytes)}
  \\[-1pt]
  \bnfmid
  & \hashE{\expe} 
  && \text{hash}
  \\[-1pt]
  \bnfmid
  & \versig{\expe}{\expei}
  && \text{signature verification}
\end{align*}

Besides constants $\valV$, basic arithmetic/logical operators, and conditionals,  
scripts can access the elements of a sequence ($\seqat{\expe}{n}$),
and the sequence of witnesses of the redeeming transaction ($\rtx.\txWit{}$);
further, they can compute the size $\sizeE{\expe}$ of a bitstring
and its hash $\hashE{\expe}$.
The script $\versig{\expe}{\expei}$
evaluates to $1$ if the signature resulting from the evaluation of $\expei$
is verified against the public key resulting from the evaluation of~$\expe$,
and $0$ otherwise.%
\footnote{Multi-signature verification is supported by Bitcoin scripts, but immaterial for our technical development: therefore, we omit it.}
For all signatures, the signed message is the redeeming transaction (except its witnesses).%

The semantics of scripts is in~\Cref{fig:bitcoin:scripts:sem}.
The script evaluation function $\sem[\txT,i]{\cdot}$ takes two additional parameters:
$\txT$ is the redeeming transaction, and
$i$ is the index of the redeeming input/witness.
The result of the semantics can be an integer, a bitstring, or a sequence of integers/bitstrings.
We denote with $\hashSem{}$ a public hash function, with $\sizeF{\valV}$ the size (in bytes) 
of an integer $\valV$,
and with $\txvername$ a signature verification function 
(the definition of these semantic operators is standard, see \eg~\cite{bitcointxm}).
The semantics of a script can be undefined, \eg when accessing an element of a non-sequence:
we denote this case as $\bot$. 
All the operators are \emph{strict}
\ie they evaluate to~$\bot$ if some of their operands is~$\bot$. 
We use standard syntactic sugar for scripts, \eg:
\begin{inlinelist}
\item $\false \eqdef 0$,
\item $\true \eqdef 1$, 
\item $\expe \andE \expei \eqdef \ifE{\expe}{\expei}{\false}$,
\item $\expe \orE \expei \eqdef \ifE{\expe}{\true}{\expei}$, and finally
\item $\notE \expe \eqdef \ifE{\expe}{\false}{\true}$.
\end{inlinelist}

\begin{figure*}[t]
  \small
  \[
  \begin{array}{c}
    \sem[\txT,i]{\valV} 
    = \valV 
    \qquad
    \sem[\txT,i]{\expe \circ \expei} 
    = \sem[\txT,i]{\expe} \circ_\bot \sem[\txT,i]{\expei}
    \qquad
    \sem[\txT,i]{\ifE{\expe[0]}{\expe[1]}{\expe[2]}} 
    = \ifSem{\sem[\txT,i]{\expe[0]}}{\sem[\txT,i]{\expe[1]}}{\sem[\txT,i]{\expe[2]}}
    \\[8pt]
    \sem[\txT,i]{\seqat{\expe}{n}} = \valV[n] \;\; 
    \text{if $\sem[\txT,i]{\expe} = \valV[1] \cdots \valV[k]$ ($1 \leq n \leq k$)}
    \qquad
    \sem[\txT,i]{\rtx.\txWit{}} = \txT.\txWit[i]{}
    \\[8pt]
    \sem[\txT,i]{\sizeE{\expe}} 
    = \sizeF{\sem[\txT,i]{\expe}}
    \qquad
    \sem[\txT,i]{\hashE{\expe}} 
    = \hashSem{\sem[\txT,i]{\expe}}
    \qquad
    \sem[\txT,i]{\versig{\expe}{\expei}}
    = \txver{\scriptsize \sem[\txT,i]{\expe}}{\sem[\txT,i]{\expei}}{}{}{\txT}{i}
  \end{array}
  \]
  \vspace{-5pt}
  \caption{Semantics of Bitcoin scripts.}
  \label{fig:bitcoin:scripts:sem}
\end{figure*}

\paragraph*{Blockchains}

A blockchain $\bcB$ is a finite sequence $\txT[0] \cdots \txT[n]$,
where $\txT[0]$ is the only coinbase transaction (\ie, $\txT[0].\txIn{} = \bot$).
%
We say that the transaction output $(\txT[i],j)$ is \emph{spent} in $\bcB$
iff there exists some $\txT[i']$ in $\bcB$ (with $i' > i$) 
and some $j'$ such that $\txT[i'].\txIn[j']{} = (\txT[i],j)$.
The \emph{unspent transaction outputs} of $\bcB$, written $\utxo{\bcB}$,
is the set of transaction outputs $(\txT[i],j)$ which are unspent in $\bcB$.
%
A transaction $\txT$ is a \emph{valid extension} of $\bcB = \txT[0] \cdots \txT[n]$
whenever the following conditions hold:
\begin{enumerate}
\item for each input $i$ of $\txT$, if $\txT.\txIn[i]{} = (\txTi,j)$ then:
  \begin{itemize}
  \item \label{consistent-update:match} 
    $\txTi = \txT[h]$, for some $h < n$ (\ie, $\txTi$ is in $\bcB$);
  \item \label{consistent-update:unspent}
    the output $(\txTi,j)$ is not spent in~$\bcB$;
  \item \label{consistent-update:output} 
    $\sem[{\txT},i]{(\txTi,j).\txscript} = \valV \neq 0$;
  \end{itemize}
\item \label{consistent-update:value}
  the sum of the amounts of the inputs of $\txT$ is greater or equal
  to the sum of the amount of its outputs.
\end{enumerate}

The Bitcoin consensus protocol ensures that 
each transaction $\txT[i]$ in the blockchain is valid
with respect to the sequence of past transactions $\txT[0] \cdots \txT[i-1]$.
The difference between the amount of inputs and that of outputs of transactions
is the \emph{fee} paid to miners who participate in the consensus protocol.

\section{Neighbourhood covenants}
\label{sec:covenants}

To extend Bitcoin with neighbourhood covenants, we amend the model of pure Bitcoin  
in the previous~\namecref{sec:bitcoin} as follows:
\begin{itemize}
\item in transactions, we add a field to outputs, making them records of the form 
  \mbox{$\{ \txarg:\ArgA,\, \txscript:\expe,\, \txval:\valV \}$},
  where $\ArgA$ is a sequence of values;
  Intuitively, this extra element can be used to encode a state within transactions;
\item in scripts, we add operators to access all the outputs of the redeeming transaction,
  and a relevant subset of those of the sibling and parent transactions
  (by contrast, pure Bitcoin scripts can only access the redeeming transaction, and only as a whole,
  to verify it against a signature);
\item in scripts, we add operators for covenants. 
\end{itemize}

We now formalize our Bitcoin extension. 
We use $\txo$ to refer to the following transaction outputs:
\begin{align*}
  \txo \;
  \bnfdef \;\;
  & \phantom{\bnfmid\;} \rtxo{}{\expe}
  && \text{output of the redeeming tx}
  \\[-1pt]
  & \bnfmid \;
    \stxo{}{\expe}
  && \text{output of a sibling tx}
  \\[-1pt]
  & \bnfmid \;
    \ptxo{}{\expe}
  && \text{output of a parent tx}
\end{align*}

More precisely, consider the case where a transaction output $(\txTi,j)$ is redeemed
by a transaction $\txT$, through its $i$-th input.
When used within the script of $(\txTi,j)$:
$\rtxo{}{n}$ refers to the $n$-th output of $\txT$;
$\stxo{}{n}$ refers to the output redeemed by the $n$-th input of $\txT$;
$\ptxo{}{n}$ refers to the output redeemed by the $n$-th input of $\txTi$.
In~\Cref{fig:covenants:txo}, we exemplify these outputs in relation to the transaction $\txT[c]$.
The semantics of $\txo$ is defined in the first line of~\Cref{fig:covenants:sem};
its result is a pair $(\txT,n)$.

We extend scripts as follows, where $\txf \in \setenum{\txarg, \txval}$:
\begin{align*}
  \expe \; 
  \bnfdef \;\;
  \cdots 
  & \bnfmid \;
    \txof{\txo}{\txf}
  && \text{field of a tx output}
  \\[-2pt]
  & \bnfmid \;
    \verscript{\expe}{\txo}
  && \text{basic covenant}
  \\[-2pt]
  & \bnfmid \;
    \verrec{\txo}
  && \text{recursive covenant}
  \\[-2pt]
  & \bnfmid \;
    \inidx
  && \text{index of redeeming tx input}
  \\[-2pt]
  & \bnfmid \;
    \outidx
  && \text{index of redeemed tx output}
  \\[-2pt]
  & \bnfmid \;
    \inlen{\txo}
  && \text{number of inputs}
  \\[-2pt]
  & \bnfmid \;
    \outlen{\txo}
  && \text{number of outputs}
  \\[-2pt]
  & \bnfmid \;
    \txid{\txo}
  && \text{hash of (tx,output)}
\end{align*}

The script $\txof{\txo}{\txf}$ gives access to the field $\txf$ of a transaction output $\txo$
(where $\txf$ is either $\txarg$ or $\txval$).
The basic covenant $\verscript{\expe}{\txo}$ checks that the script in the transaction output $\txo$
is syntactically equal to $\expe$ (note that $\expe$ is not evaluated).
The ``recursive'' covenant $\verrec{\txo}$ checks that the script in $\txo$
is syntactically equal to script which is currently being evaluated.
The operators $\inidx$ and $\outidx$ evaluate, respectively,
to the index of the redeeming input and redeemed output.
We call \emph{current transaction output} ($\ctxo{}$) the transaction output which includes
the script which is currently being evaluated, \ie:
\[
\ctxo{} \; \eqdef \; \stxo{}{\inidx}
\]
The scripts $\inlen{\txo}$ and $\outlen{\txo}$ evaluate, respectively,
to the number of inputs and to the number of outputs of the transaction containing $\txo$.
Finally, $\txid{\txo}$ evaluates to a unique identifier of the transaction output $\txo$.

\begin{example}
  \label{ex:covenants:txo}
  Consider the transactions in~\Cref{fig:covenants:txo}.
  The script in $\txT[c].\txOut[1]{}$ checks that 
  \begin{inlinelist}
  \item the $\txarg$ field of $\ctxo{}$, \ie the same output which contains the script, equals to $n_c$;
  \item the $\txarg$ field of $\ptxo{}{1}$, \ie the parent transaction output redeemed by
    $\txT[c].\txIn[1]{}$, equals to $n_p$;
  \item the $\txarg$ field of $\rtxo{}{3}$, \ie the third output of the redeeming transaction $\txT[r]$, 
    equals to $n_r$;
  \item the $\txarg$ field of $\stxo{}{2}$, \ie the sibling transaction output redeemed by 
    $\txT[r].\txIn[2]{}$, equals to $n_s$.
  \end{inlinelist}
  Note that, in general, any script used in $\txT[c]$ can not access
  the parents of $\txT[p]$ and those of $\txT[s]$ --- and in general all the transactions
  which are farther than those shown in the figure.
  \qedex
\end{example}

\newcommand{\txFigPtx}{%
  \begin{tabular}[t]{|l|}
    \hline
    \multicolumn{1}{|c|}{$\txT[p]$} \\
    \hline
    \txIn{$\cdots$} \\
    \txOut[1]{}: $\{ \txarg:n_p, \cdots \} \textcolor{darkgray}{\text{ // } \ptxo{}{1}}$ \\ 
    \hline
  \end{tabular}
}

\newcommand{\txFigPtxi}{%
  \begin{tabular}[t]{|l|}
    \hline
    \multicolumn{1}{|c|}{$\txTi[p]$} \\
    \hline
    \txIn{$\cdots$} \\
    \txOut[1]{}: $\{ \cdots \} \textcolor{darkgray}{\text{ // } \ptxo{}{2}}$ \\ 
    \hline
  \end{tabular}
}

\newcommand{\txFigCtx}{%
  \begin{tabular}[t]{|l|}
    \hline
    \multicolumn{1}{|c|}{$\txT[c]$} \\
    \hline
    \txIn{$(\txT[p], 1) \; (\txTi[p], 1)$} \\
    \txOut[1]{}: $\{$ \\
    \quad\txarg: $n_c$ \\
    \quad\txscript: $\txof{\ctxo{}}{\txarg}=n_c$ \\
    \quad\hspace{5pt} $\andE \txof{\ptxo{}{1}}{\txarg}=n_p$ \\
    \quad\hspace{5pt} $\andE \txof{\rtxo{}{3}}{\txarg}=n_r$ \\
    \quad\hspace{5pt} $\andE \txof{\stxo{}{2}}{\txarg}=n_s$ \\
    $\}$ \\[2pt]
    \hline
  \end{tabular}
}

\newcommand{\txFigStx}{%
  \begin{tabular}[t]{|l|}
    \hline
    \multicolumn{1}{|c|}{$\txT[s]$} \\
    \hline
    \txIn{$\cdots$} \\
    \txOut[1]{}: $\{ \txarg:n_s, \cdots \} \textcolor{darkgray}{\text{ // } \stxo{}{2}}$ \\ 
    \txOut[2]{}: $\{ \cdots \} \hspace{8pt} \textcolor{darkgray}{\text{ // not accessible }}$ \\ 
    \hline
  \end{tabular}
}

\newcommand{\txFigRtx}{%
  \begin{tabular}[t]{|l|}
    \hline
    \multicolumn{1}{|c|}{$\txT[r]$} \\
    \hline
    \txIn{$(\txT[c], 1) \; (\txT[s], 1)$} \\
    \txOut[1]{}: $\cdots \hspace{4pt} \textcolor{darkgray}{\text{ // } \rtxo{}{1}}$ \\
    \txOut[2]{}: $\cdots \hspace{4pt} \textcolor{darkgray}{\text{ // } \rtxo{}{2}}$ \\
    \txOut[3]{}: $\{ \txarg:n_r, \cdots \} \textcolor{darkgray}{\text{ // } \rtxo{}{3}}$ \\ 
    \hline
  \end{tabular}
}

\begin{figure}[t]
  \resizebox{\columnwidth}{!}{
    \begin{tikzpicture} 
      \node at (0,3) {\txFigPtx};
      \node at (4.75,3) {\txFigPtxi};
      \node at (-0.45,0.15) {\txFigCtx};
      \node at (4.5,1.25) {\txFigStx};
      \node at (4.5,-0.85) {\txFigRtx};
    \end{tikzpicture}
  } 
  \caption{Accessing transaction outputs through a script.}
  \label{fig:covenants:txo}
\end{figure}

\begin{figure*}[t]
  \small
  \[
  \begin{array}{c}
    \sem[\txT,i]{\rtxo{}{\expe}} 
    = (\txT, \sem[\txT,i]{\expe})
    \qquad
    \sem[\txT,i]{\stxo{}{\expe}}
    = \txT.\txIn[{\sem[\txT,i]{\expe}}]{}
    \qquad
    \sem[\txT,i]{\ptxo{}{\expe}}
    = \txTi.\txIn[{\sem[\txT,i]{\expe}}]{}
    \; \text{ if } \txT.\txIn[i]{} = (\txTi,j)
    \\[8pt]
    \sem[\txT,i]{\txof{\txo}{\txf}}
    = \sem[\txT,i]{\txo}.{\txf}
    \qquad
    \sem[\txT,i]{\verscript{\txo}{\expe}} 
    = \expe \equiv \sem[\txT,i]{\txo}.\txscript
    \qquad
    \sem[\txT,i]{\verrec{\txo}} 
    = \txT.\txIn[i]{}.\txscript \equiv  \sem[\txT,i]{\txo}.\txscript
    \qquad
    \sem[\txT,i]{\txid{\txo}}
    = \hashSem{\sem[\txT,i]{\txo}}
    \\[8pt]
    \sem[\txT,i]{\outidx} 
    = j    
    \;\; \text{if } \txT.\txIn[i]{} = (\txT,j)
    \qquad
    \sem[\txT,i]{\inidx} 
    = i
    \qquad
    \sem[\txT,i]{\outlen{\txo}} 
    = |\txTi.\txOut{}| 
    \qquad
    \sem[\txT,i]{\inlen{\txo}} 
    = |\txTi.\txIn{}|
    \; \text{ if } \sem[\txT,i]{\txo} = (\txTi,j)
  \end{array}
  \]
  \caption{Semantics of neighbourhood covenants (extending~\Cref{fig:bitcoin:scripts:sem}).}
  \label{fig:covenants:sem}
\end{figure*}

\Cref{fig:covenants:sem} defines the semantics of extended scripts.
As in~\Cref{sec:bitcoin}, the function $\sem{\cdot}$ takes as parameters
the redeeming transaction $\txT$ and the index $i$ of the redeeming input.
We denote with $\equiv$ syntactic equality between two scripts,
\ie $\expe \equiv \expei$ is $1$ when $\expe$ and $\expei$ are exactly the same, $0$ otherwise.

\paragraph*{Turing completeness}

Our neighbourhood covenants make Bitcoin Turing-complete. 
Indeed, this is also true without exploiting the 
$\stxo{}{}$ and $\ptxo{}{}$ operators,
\ie by only using covenants between the current and the redeeming transaction.
To prove this, we describe how to simulate in the extended Bitcoin 
any counter machine~\cite{FMR68}, a well-known Turing-complete computational model.
A \emph{counter machine} is a pair $(n,s)$, 
where $n \in \Nat$ is the number of integer registers of the machine, and $s$ is a sequence of instructions. 
Instructions are the following: 
\mbox{${\sf inc}\ i$} increments register~$i$, 
\mbox{${\sf dec}\ i$} decrements it, 
\mbox{${\sf zero}\ i$} sets it to zero, 
\mbox{${\sf if}\ i\neq 0$ ${\sf goto}\ j$} conditionally jumps to instruction $j$ when register $i$ is not zero, 
\mbox{${\sf halt}$} terminates the machine.
The state of a counter machine $(n,s)$ is a tuple $(v_1,\ldots,v_n, p)$ 
where each $v_i$ represents the current value of register $i$, 
and $p$ is the number of the next instruction to execute (\ie, the program counter).
To exploit the currency transfer capabilities of Bitcoin, 
we slightly extend the counter machine model, 
by requiring that the machine has an initial $\BTC$ balance which, upon termination, 
is transferred to the user $\pmvA$ if the content of the first register is $0$, or to $\pmvB$ otherwise.
We call this extended model \emph{UTXO-counter machine}.

\begin{thm}
  \label{th:covenants:turing-completeness}
  Neighbourhood covenants can simulate any UTXO-counter machine.
  Hence, they are Turing-complete.
\end{thm}
\begin{proof}(\emph{sketch})
  We represent the state of the counter machine as a single transaction having one output.
  We simulate an execution step by appending a new transaction, 
  which redeems the output representing the old state, 
  and transfers its balance to a new output representing the new state. 
  This is done until the machine halts, at which point 
  we transfer its balance to the user determined by the final registers state.
  We remark that this simulation is made possible by the use of unbounded integers in our model, while Bitcoin only supports 32-bit integers.

  We represent the machine state in the $\txarg$ field of the transaction output $\txo$, 
  storing it as a sequence of integers. 
  As a shorthand, we write $\txo.\txr{i}$ for $\txo.\txarg.i$, and $\txo.\txp$  for $\txo.\txarg.(n+1)$.
  To simulate the execution steps of the machine, 
  we define the script $\CM$, 
  which checks that the new transaction $\txT$ indeed represents the next state. 
  More in detail, we first check whether the instruction in $s$ at position $\ctxo{\txp}$ 
  is \mbox{$\sf halt$}, 
  in which case we require that $\txT$ distributes the balance to users in the intended manner, 
  depending on the current state.
  When $\ctxo{\txp}$ points to any other instruction, we start by requiring that $\txT$ has only one output, 
  the same balance, and the same script. 
  We use a recursive covenant $\verrec{}$ on the redeeming transaction to ensure the last part.
  Then, we check that the new state in $\txT.\txOut[1]{}.\txarg$ agrees with the counter machine semantics, 
  by cases on the instruction pointed to by $\txo.\txp$. 
  If the instruction is \mbox{${\sf inc}\ i$}, then we require 
  $\rtxo{}{1}.\txr{i} = \ctxo{}.\txr{i}+1$, $\rtxo{}{1}.\txr{k} = \ctxo{}.\txr{k}$ for all $k\neq i$, 
  and $\rtxo{}{1}.\txp = \ctxo{}.\txp + 1$. 
  The \mbox{${\sf dec}\ i$} and \mbox{${\sf zero}\ i$} cases are analogous.
  For the instruction \mbox{${\sf if}\ i\neq 0$ ${\sf goto}\ j$}, 
  we check the value of $\ctxo{}.\txr{i}$: 
  if nonzero, we require that $\rtxo{}{1}.\txp = j$,
  otherwise that $\rtxo{}{1}.\txp = \ctxo{}.\txp + 1$. 
  In both cases, we require that $\rtxo{}{1}.\txr{k} = \ctxo{}.\txr{k}$ for all $k$.

  Users start the simulation by appending to the blockchain a transaction $\txT[0]$ 
  having one output with the desired value, 
  the script $\CM$, and a sequence of \mbox{$n+1$} zeros as $\txarg$. 
  After that, the balance is effectively locked inside the transaction, 
  and the only way to transfer it back to the users is to simulate all the steps of the machine, 
  until it halts. 
  So, our simulated execution is a form of \emph{secure multiparty computation}~\cite{Andrychowicz16cacm,Yao86focs,Goldreich87stoc}.
\end{proof}

Although, for simplicity, we use UTXO-counter machines just to transfer funds to $\pmvA$ or $\pmvB$ upon termination,
it would be easy to generalise the computational model and simulation technique
to encompass \emph{interactive} computations, 
which at run-time can receive inputs and perform currency transfers.
Doing so, we can execute arbitrary smart contracts, 
with the same expressiveness of Turing-complete smart contracts platforms.
In principle, we could craft a smart contract which implements user-defined tokens
in an account-based fashion: this contract would record the balance of the tokens of all users, 
and execute token actions.
However, in practice this construction would be highly inefficient: 
performing a single token transfer would require to append a large number of transactions,
and to pay the related fees.

To improve the efficiency of tokens, 
instead of limiting to covenants on the redeeming transaction, as in the simulation above, 
we could also use covenants on the parent and sibling transactions.
In~\Cref{sec:tokens:c} we show an efficient implementation of tokens,
which fully exploits neighbourhood covenants.

  \paragraph*{Implementing neighbourhood covenants}
\label{sec:implementation}

To extend Bitcoin with neighborhood covenants, only small changes to the script language are needed.
Currently, scripts can only access the redeeming transaction, and only for signature verification.
To enable covenants, scripts need to access the fields of the redeeming transaction,
and those of the parent and the \emph{sibling} transaction outputs.
First, the UTXO data structure must be extended to record the parents of unspent outputs.
Full nodes could simply access these transactions from their identifiers.
Lightweight nodes, \ie Bitcoin clients with limited resources that store the UTXO instead of the whole blockchain, must also store parents beside the UTXO;
when all the children of a transaction are spent, the parent can be deleted.

To implement the covenants $\verrec{}$ and $\verscript{}{}$, 
the Bitcoin script language must be extended with new opcodes.
A former covenant-enabling opcode is the ${\small\sf CheckOutputVerify}$ of~\cite{Moser16bw},
which uses placeholders to represent variable parts of the script
(\eg, $\versig{\texttt{<\textit{pubKey}>}}{\rtx.\txWit{}}$).
However, its implementation requires string substitutions at run-time 
to insert the wanted values in the script before checking script equality.
Instead, our neighbourhood covenants can be implemented more efficiently.
In our covenants, we can use exactly the same script along a chain of transactions,
relying on the $\txarg$ sequence to record state updates.
The script can access the state through the operator $\txof{\txo}{\txarg}$,
which require suitable opcodes.
Since the script is fixed, nodes do not have to perform string substitutions,
and checking script equality can be efficiently performed by comparing their hashes. 

Adding the $\txarg$ field to outputs does not require to
alter the structure of pure Bitcoin transactions.
Indeed, the $\txarg$ values can be stored at the beginning of the script
as push operations on the alternative stack,
and copied to the main stack when the script refers to them,
using the same technique used in~\cite{balzac}. 
When hashing the scripts for comparison, we discard these $\txarg$ values:
this just requires to skip the prefix of the script comprising
all the {\small\sf push} to the alternative stack.

\section{Implementing tokens in Bitcoin}
\label{sec:tokens:c}

In this~\namecref{sec:tokens:c} we show how to implement token actions in Bitcoin.
To this purpose, we define a \emph{computational model}, 
which describes the interactions of users who exchange messages 
and append transactions to the Bitcoin blockchain.
A \emph{computational run} $\runC$ is a sequence of bitstrings $\labC$, 
each of which encodes one of the following actions:
\begin{inlinelist}
\item \mbox{$\pmvA \rightarrow *:m$}, denoting the broadcast of a bitstring $m$;
\item $\txT$, denoting the appending of a transaction $\txT$ to the blockchain.
\end{inlinelist}
A computational run always starts from a coinbase transaction $\txT[0]$.
By extracting the transactions from a run $\runC$, 
we obtain a blockchain $\bcB[\runC]$.

We relate the computational and the symbolic models through a relation
between symbolic runs $\runS$ and computational runs $\runC$:
intuitively, $\runS$ is \emph{coherent} with $\runC$ 
when each step in $\runS$ is simulated by a step in~$\runC$.
In the rest of this~\namecref{sec:tokens:c} we provide the main intuition about 
the coherence relation, 
\iftoggle{arxiv}{postponing the full details to~\appendixname~\ref{app:coherence}.}
{relegating the full details to~\cite{BLZ20arxiv}.}
The coherence relation, denoted as~$\!\coherRel{}{}{}\!$, is parameterized over two injective functions
$\txMapC$ and $\tkMapC$ which track, respectively,
the \emph{names} $\varX$ 
and the \emph{tokens} $\tokT$ occurring in $\confS{\runS}$,
mapping them to transaction outputs $(\txT,i)$, where $\txT$ occurs in $\runC$.
We abbreviate $\coherRel{\runS}{\runC}{\txMapC,\tkMapC}$ as $\coherRel{\runS}{\runC}{}$.

We simulate each symbolic token action in \Cref{fig:tokens-s:semantics}
by appending a suitable transaction to the blockchain.
We use the $\txarg$ part of transaction outputs to record the token data:
\begin{enumerate}[1.]
\item $\tkop$ is the action implemented by the transaction: 
  $0 = \genOp$, $1 = \burnOp$, $2 = \splitOp$, $3 = \joinOp$, $4 = \exchangeOp$, $5 = \giveOp$;
\item $\tkown$ is the (public key of the) user who owns of the token units controlled by the tx output; 
\item $\tkval$ is the number of units controlled by the tx output; 
\item $\tkid$ is the unique token identifier.
\end{enumerate}

We implement the token actions as a single script $\tokScript$, 
which uses a switch on the $\tkop$ value to jump to the first instruction of the requested action.
Since the script is quite complex, we present independently 
the parts corresponding to each action
(we display the complete script in~\Cref{fig:full-token:script}).
All the transactions implementing token actions use \emph{exactly} the same script,
which we preserve throughout executions by using recursive covenants.
To improve readability, 
we refer to the elements of the $\txarg$ sequence by name rather than by index, 
\eg we write $\txof{\txo}{\tkop}$ rather than $\txof{\txo}{\seqat{\txarg}{1}}$.



\paragraph*{Gen}

\newcommand{\ftokengen}{{\sf G}}

We implement the $\genOp$ action by the following script:
\begin{lstlisting}[language=balzac,numbers=left,numbersep=5pt,xleftmargin=10pt,classoffset=1,morekeywords={},classoffset=2,morekeywords={FundsA,Commit,Reveal},framexbottommargin=0pt]
if not verrec(ptxo(1))         // ctxo is a gen
then ctxo.tkid = txid(ptxo(1)) // token id
     and ptxo(1).val = 0       // spent txo has 0 BTC
     and outlen(ctxo) = 1      // gen has 1 out
     and ctxo.tkval > 0        // positive token val
else ... // the other branches must preserve token id
\end{lstlisting}

\noindent
Recall that $\genOp$ produces a symbolic step of the form:
\[
\confDep[\varX]{\pmvA}{0:\BTC} 
\xrightarrow{\actGen{\varX}{\valV}} 
\confDep[\varY]{\pmvA}{\valV:\tokT} 
\tag*{$(v>0)$}
\]
To translate this symbolic action into a computational one,
we must spend a transaction output corresponding to $\confDep[\varX]{\pmvA}{0:\BTC} $,
and produce a fresh output corresponding to $\confDep[\varY]{\pmvA}{\valV:\tokT}$.
Assuming that the deposit $\varX$ corresponds to an unspent transaction output $(\txTi,1)$ on the blockchain,
this requires to append a transaction $\txT$ redeeming $(\txTi,1)$, and ensuring that:
\begin{enumerate}
\item the parent transaction output $(\txTi,1)$ is not a token deposit, but just a plain $\BTC$ deposit
  (line~\lineno{1});
\item $\tkid$ is the identifier of the parent tx output (line~\lineno{2}).
  This corresponds to identifying the fresh name $\tokT$ 
  with the deposit name $\varX$ of the redeemed $\BTC$ deposit;
\item $0 \BTC$ are redeemed from the parent transaction (line~\lineno{3});
\item $\txT$ has exactly one output (line~\lineno{4});
\item $\tkval$ is positive (line~\lineno{5}), corresponding to 
  the constraint $\valV > 0$ in the symbolic semantics.
\end{enumerate}

Notice that the first time the script is evaluated is when \emph{redeeming} $\txT$ (not when appending it).
At that time, $\ctxo{}$ will evaluate to $(\txT,1)$, and $\ptxo{}{1}$ to $(\txTi,1)$.
The script ensures that, when $\txT$ is redeemed, its $\tkid$ will contain a unique identifier
of the token.
Crucially, the scripts which implement the other token actions will preserve this identifier
in the $\tkid$ parameter.
This identifier is essential to guarantee the correctness of the $\joinOp$ and $\exchangeOp$ actions.

\paragraph*{Burn}
\label{ex:burn}

We implement the $\burnOp$ action by the script: \\
\vbox{\begin{lstlisting}[language=balzac,numbers=left,numbersep=5pt,xleftmargin=10pt,classoffset=1,morekeywords={},classoffset=2,morekeywords={},framexbottommargin=0pt]
versig(ctxo.owner,rtx.wit) and // check owner
verscr(false,rtxo(1)) and     // make rtx unspendable
inlen(rtxo(1)) = 1 and        // rtx has 1 in
outlen(rtxo(1)) = 1           // rtx has 1 out
\end{lstlisting}}

\noindent
Recall that $\burnOp$ produces a symbolic step:
\[
\big( \mmid_{i \in 1..n} \confDep[x_i]{\pmvA[i]}{\valV[i]:\tokT[i]} \big) 
\xrightarrow{\actBurn{x_1 \cdots x_n}{y}}
\gnil
\]
There are two cases, according to whether we are burning a single token deposit, 
or one or more $\BTC$ deposits.
In the first case, assuming that the computational counterpart of $\varX[1]$ is the output $(\txTi,1)$,
the corresponding computational step is to append a transaction $\txT$ redeeming $(\txTi,1)$.
The witness of $\txT$ must carry a signature of the owner $\pmvA$, 
and its output script is $\false$, making it unspendable.
In the second case, it suffices to append a transaction $\txT$ which redeems all the transaction outputs
corresponding to $\varX[1], \ldots, \varX[n]$, and has a $\false$ script.

\paragraph*{Split}
\label{ex:split}

We implement the $\splitOp$ action by the script: 
  \begin{lstlisting}[language=balzac,numbers=left,numbersep=5pt,xleftmargin=10pt,classoffset=1,morekeywords={},classoffset=2,morekeywords={},framexbottommargin=0pt]
versig(ctxo.owner,rtx.wit)             // check owner
and verrec(rtxo(1))               // covenants on rtx
and verrec(rtxo(2))  
and inlen(rtxo(1)) = 1                // rtx has 1 in
and outlen(rtxo(1)) = 2             // rtx has 2 outs
and rtxo(1).tkval >= 0        // positive token value
and rtxo(2).tkval >= 0          
and rtxo(1).owner = ctxo.owner      // preserve owner
and rtxo(1).tkid = ctxo.tkid         // preserve tkid
and rtxo(2).tkid = ctxo.tkid     
and rtxo(1).tkval + rtxo(2).tkval = ctxo.tkval
  \end{lstlisting}

Recall that $\splitOp$ produces a symbolic step of the form:
\[
\confDep[\varX]{\pmvA}{(\valV+\valVi):\tokT} 
\xrightarrow{\actSplit{\varX}{\valV}{\pmvB}} 
\confDep[\varY]{\pmvA}{\valV:\tokT} 
\mid
\confDep[\varZ]{\pmvB}{\valVi:\tokT} 
\]

Assuming that $\varX$ corresponds to an unspent transaction output $(\txTi,1)$,
performing this step in Bitcoin requires to append a transaction $\txT$ redeeming $(\txTi,1)$, 
and ensuring that:
\begin{enumerate}
\item the witness of $\txT$ carries a signature of the owner (line~\lineno{1});
\item $\txT$ has only one input and two outputs, both containing the same script of $(\txTi,1)$
  (line~\lineno{2}-\lineno{5});
\item the $\tkval$ of $\txT$'outputs $\rtxo{}{1}$ and $\rtxo{}{2}$ are $\geq 0$
  (line~\lineno{6}-\lineno{7}), corresponding to  
  the precondition $\valV, \valVi \geq 0$ in \nrule{[Split]};
\item $\tkid$ of $\txT$'s outputs is the same of $(\txTi,1)$ (line~\lineno{9}-\lineno{10});
\item the sum of token values $\tkval$ of the outputs of $\txTi$ is equal to the token value of $(\txTi,1)$ 
  (line-\lineno{11}).
\end{enumerate}
\paragraph*{Join}

We implement the $\joinOp$ action by the script: 
\begin{lstlisting}[language=balzac,numbers=left,numbersep=5pt,xleftmargin=10pt,classoffset=1,morekeywords={},classoffset=2,morekeywords={},framexbottommargin=0pt]
inlen(rtxo(1)) = 2                   // rtx has 2 ins
and outlen(rtxo(1)) = 1              // rtx has 1 out
and verrec(rtxo(1))               // coventant on rtx
and verrec(stxo(2))       // covenants on both inputs
and verrec(stxo(1))        
and ctxo.tkid = rtxo(1).tkid     // preserve token id
and versig(ctxo.owner, rtx.wit) // check sig of owner
and rtxo(1).tkval = stxo(1).tkval + stxo(2).tkval    
\end{lstlisting}

Recall that $\joinOp$ produces a symbolic step of the form:
\[
\confDep[\varX]{\pmvA}{\valV:\tokT} 
\mid
\confDep[\varY]{\pmvB}{\valVi:\tokT} 
\; \xrightarrow{\actJoin{\varX}{\varY}{\pmvC}} \;
\confDep[\varZ]{\pmvC}{(\valV+\valVi):\tokT} 
\]

Assume that $\varX$ and $\varY$ correspond to the
unspent transaction outputs $(\txTi,1)$ and  $(\txTii,1)$.
To perform the corresponding computational step
we append a transaction $\txT$ redeeming $(\txTi,1)$ and $(\txTii,1)$, and ensuring that,
for both inputs:
\begin{enumerate}
\item $\txT$ has two inputs and one output, containing the same script of $(\txTi,1)$ and $(\txTii,1)$
  (line~\lineno{1}-\lineno{5});
\item the token identifier $\tkid$ of the output of $\txTi$ ($\rtxo{}{1}$)
  is the same of $(\txTi,1)$ (line~\lineno{6});
\item the witness of $\txT$ carries a signature of the owner (line~\lineno{7});
\item the sum of token values $\tkval$ of both inputs
  is equal to the token value of $(\txT,1)$ 
  (line~\lineno{8}).
\end{enumerate}

Note that the script in $(\txTi,1)$ ensures that the one in $(\txTii,1)$ is the same, and vice-versa.
In this way, we prevent joining tokens with bitcoins.
To also prevent joining tokens of different type, 
the script checks that the $\tkid$ of the current transaction is the same as 
the one of the first output of the redeeming transaction.
This is done by both inputs.
In other words, 
$\stxo{\tkid}{1} = \rtxo{\tkid}{1}$ and
$\stxo{\tkid}{2} = \rtxo{\tkid}{1}$,
implying that 
$\stxo{\tkid}{1} = \stxo{\tkid}{2}$.

\paragraph*{Exchange}

We implement the $\exchangeOp$ action by the script: 
\begin{lstlisting}[language=balzac,numbers=left,numbersep=5pt,xleftmargin=10pt,classoffset=1,morekeywords={},classoffset=2,morekeywords={FundsA,Commit,Reveal},framexbottommargin=0pt]
inlen(rtxo(1)) = 2                   // rtx has 2 ins
and outlen(rtxo(1)) = 2             // rtx has 2 outs
and verrec(stxo(1))            // covenant on input 1
and verrec(rtxo(1))             // covenant on rtx(1)
and versig(ctxo.owner, rtx.wit)       // check  owner
and rtxo(1).owner = stxo(2).owner   // exchange owner
and rtxo(2).owner = stxo(1).owner  
and rtxo(1).tkval = stxo(1).tkval   // preserve value 
and rtxo(1).tkid = stxo(1).tkid      // preserve tkid 
if verrec(stxo(2)) then       // exchange token/token
      verrec(rtxo(2))           // covenant on rtx(2)
  and rtxo(2).tkval = stxo(2).tkval // preserve tkval
  and rtxo(2).tkid = stxo(2).tkid    // preserve tkid
else                            // exchange token/BTC
      verscr(versig(ctxo.owner, rtx.wit), rtxo(2))
  and rtxo(2).val = stxo(2).val       // preserve BTC
\end{lstlisting}

Recall that the symbolic $\exchangeOp$ step has the form:
\[
\confDep[x]{\pmvA}{\valV:\tokT} \mid \confDep[y]{\pmvB}{\valVi:\tokTi} 
\xrightarrow{\actExchange{x}{y}} 
\confDep[x']{\pmvA}{\valVi:\tokTi} \mid \confDep[y']{\pmvB}{\valV:\tokT}
\]
where $\tokT$ must be a token, while $\tokTi$ is either a $\BTC$ or a token.

Assume that $\varX$ and $\varY$ correspond to the unspent transaction outputs $(\txTi,1)$ and  $(\txTii,1)$.
To perform the corresponding computational step
we append a transaction $\txT$ redeeming $(\txTi,1)$ and $(\txTii,1)$, and ensuring that,
for both inputs:
\begin{enumerate}
	\item $\txT$ has two inputs and two output (lines~\lineno{1}-\lineno{2});
	\item the first input and the first output of $\txT$ must contain the same script of $(\txTi,1)$ and $(\txTii,1)$
	(lines~\lineno{3}-\lineno{4});
	\item the witness of $\txT$ carries a signature of the owner (line~\lineno{5});
	\item the owner in the first output of $\txT$ ($\rtxo{}{1}$) must be equal to the owner in the second input  $(\txTii,1)$ (line~\lineno{6});
	\item dually, the owner in the second output of $\txT$ ($\rtxo{}{2}$) must be equal to the owner in the first input  $(\txTi,1)$ (line~\lineno{7});
	\item the token value and identifier of the first output of $\txT$ ($\rtxo{}{1}$) 
	must be equal to those of $(\txTi,1)$ (line~\lineno{8}-\lineno{9}).
\end{enumerate}
Furthermore, if $\verrec{\stxo{}{2}}$ is true, \ie we are exchanging a token with a token.
In this case, we require that:
\begin{enumerate}
	\item the second output of $\txT$ must contain the same script of $(\txTi,1)$ and $(\txTii,1)$ (line~\lineno{11});
	\item the token value and identifier of the second output of $\txT$ ($\rtxo{}{2}$) 
	must be equal to those of $(\txTii,1)$ (line~\lineno{12}-\lineno{13});
\end{enumerate} 
When exchanging a token with a $\BTC$ deposit, we require:
\begin{enumerate}
	\item the script of the second output of $\txT$ ($\rtxo{}{2}$) to be $\versig{\ctxo{}.\tkown}{\rtx.\txWit{}}$ (line~\lineno{15});
	\item the value of the second output of $\txT$ 	to be equal to the value of $(\txTii,1)$ (line~\lineno{16}).
\end{enumerate}
\paragraph*{Give}

We implement the $\giveOp$ action by the script: \\
\vbox{\begin{lstlisting}[language=balzac,numbers=left,numbersep=5pt,xleftmargin=10pt,classoffset=1,morekeywords={},classoffset=2,morekeywords={},framexbottommargin=0pt]
inlen(rtxo(1)) = 1                    // rtx has 1 in
and outlen(rtxo(1)) = 1              // rtx has 1 out
and versig(ctxo.owner, rtx.wit)       // check  owner
and verrec(rtxo(1))             // covenant on rtx(1)
and rtxo(1).tkid = ctxo.tkid         // preserve tkid
and rtxo(1).tkval = ctxo.tkval      // preserve value                        
\end{lstlisting}}

\noindent
Recall that $\giveOp$ produces a symbolic step of the form:
\[
\confDep[x]{\pmvA}{\valV:\tokT}
\xrightarrow{\actGive{x}{\pmvB}} 
\confDep[y]{\pmvB}{\valV:\tokT} 
\]
Assuming that $\varX$ corresponds to an unspent transaction output $(\txTi,1)$,
performing this step in Bitcoin requires to append a transaction $\txT$ redeeming $(\txTi,1)$, 
and ensuring that:
\begin{enumerate}
\item $\txT$ has only one input and one output (lines~\lineno{1}-\lineno{2});
\item the witness of $\txT$ carries a signature of the owner (line~\lineno{3});
\item the output of $\txT$ ($\rtxo{}{1}$) contains the same script of $(\txTi,1)$
  (line~\lineno{4});
\item the token value and identifier of the first output of $\txT$ ($\rtxo{}{1}$) 
  is the same of those of $(\txTi,1)$ (lines~\lineno{5}-\lineno{6}).
\end{enumerate}


\paragraph*{Authorizations} 

Symbolic authorization steps (\Cref{fig:tokens-s:auth:semantics}) 
correspond to computational broadcasts of signatures. 
Computational users, however, can also broadcast other messages:
in particular, adversaries can broadcast any arbitrary bitstring they can compute in PPTIME. 
Coherence discards any broadcast which does not correspond to any of the symbolic steps above,
\ie, such broadcasts correspond to no symbolic steps. 
Discarding these messages does not affect the security of tokens, 
because the other computational messages 
(\ie, transactions and their signatures) 
are enough to reconstruct the symbolic run from the computational one.

\paragraph*{Other transactions}

A subtle case of coherence is that of transactions appended by \emph{dishonest} users.
To illustrate the issue, suppose that some dishonest $\pmvA[1] \cdots \pmvA[n]$ own some bitcoins,
represented in the symbolic run
as deposits $\confDep[{\varX[i]}]{\pmvA[i]}{\valV[i]:\BTC}$ for $i \in 1..n$, 
and in the computational run as transaction outputs $\txT[x_1] \cdots \txT[x_n]$,
each one redeemable with $\pmvA[i]$'s private key.
The dishonest users can sign an \emph{arbitrary} $\txTi$ 
which redeems all the $\txT[x_1] \cdots \txT[x_n]$,
and append it to the blockchain.
Crucially, $\txTi$ may not correspond to $\genOp$, $\splitOp$, $\joinOp$, $\exchangeOp$ or $\giveOp$ actions.
In this case, to obtain coherence, we simulate the appending of $\txTi$ 
with a (previously authorized) \emph{burn} of the deposits 
$\confDep[{\varX[i]}]{\pmvA[i]}{\valV[i]:\BTC}$.
In subsequent steps, coherence will ignore the descendants of $\txTi$ in the computational run,
since in general they cannot be represented symbolically.
More precisely, 
appending a transaction where none of the inputs corresponds to a symbolic deposit
results in no symbolic action.
The case of a deposit $\confDep[\varX]{\pmvA}{\valV:\tokT}$ with a user-generated token $\tokT$ is similar: 
appending $\txTi$ is simulated by a symbolic \emph{burn},
and $\txTi$ is not represented symbolically.
In all cases, a transaction which spends a symbolically-represented input
must be represented symbolically, otherwise coherence is lost.

\paragraph*{Efficiency of the implementation}

To estimate the efficiency of the implementation, 
we consider the number of cryptographic operations,
as their execution cost is an order of magnitude
greater than the other operations.  
In particular, performing $\verrec{}$ and $\verscript{}{}$ requires to compute the hash of a script
(once this is done, the cost of comparing two hashes is negligible).
This cost can be reduced by incentivizing nodes to cache scripts.
The most expensive token action is $\exchangeOp$, which, 
having two inputs, needs to verify $2$ signatures and execute at most $10$ covenants operations,
which overall require to compute at most $6$ script hashes.
If nodes cache scripts, the cost of the action is not 
dissimilar to the one required to append a standard transaction with two inputs.

Note that, even though $\tokScript$ is a non-standard script,
it could be used in a standard {\small\sf P2SH} transaction, as in~\cite{balzac},
if it did not exceed the 520-bytes limit.
Taproot \cite{bip341} would mitigate this issue:
for scripts with multiple disjoint branches, 
Taproot allows the witness of the redeeming transaction
to reveal just the needed branch.
Therefore, the 520-bytes limit would apply to branches instead of the whole script.

\section{Computational soundness}
\label{sec:computational-soundness}

An immediate consequence of the definition of coherence is the following~\namecref{lem:coher:s-to-c},
which ensures that unspent deposits in a symbolic run $\runS$ 
have a corresponding unspent transaction output in any computational run $\runC$ coherent with $\runS$.

\begin{lem}
  \label{lem:coher:s-to-c}
  Let \mbox{$\coherRel{\runS}{\runC}{}$}.
  For each deposit $\confDep[\varX]{\pmvA}{\valV:\tokT}$ in $\confS{\runS}$,
  there exists a distinct tx output $\txMapC(x)$ in $\utxo{\bcB[\runC]}$
  storing $\valV:\tokT$, and spendable through a signature of $\pmvA$.
\end{lem}

\Cref{lem:coher:c-to-s} is a sort of dual of \Cref{lem:coher:s-to-c}, 
describing how to relate computational token deposits to symbolic ones.
In general, this is complex since a computational adversary
can craft arbitrary scripts which are not representable in the symbolic model:
hence, the blockchain can contain transactions $\txT$ which do not correspond
to any symbolic deposit.
A tricky case is when, after some token $\tokT$ is minted in the symbolic world,
a computational adversary creates a descendent $\txTi$ of $\txT$ 
with the token script and $\tkid = \tokT$, effectively forging new units of $\tokT$.
Note that such forgery can not be prevented, since the
adversary can create transactions with arbitrary outputs.
However, such forged tokens are useless: this is guaranteed by \Cref{lem:coher:c-to-s}, which shows that $\txTi$ is unspendable.

\begin{lem}
  \label{lem:coher:c-to-s}
  Let \mbox{$\coherRel{\runS}{\runC}{}$}.
  Let $(\txT,i)$ be a tx output in $\utxo{\bcB[\runC]}$ storing $\valV:\tokT$,
  with $\tokT \neq \BTC$.
  If $(\txT,i).\tkid$ does not correspond to any tx output in $\bcB[\runC]$, 
  then $(\txT,i)$ is unspendable.
  Otherwise, if $(\txT,i).\tkid$ occurs in $\bcB[\runC]$,
  and $(\txT,i).\tkid$ is equal to $\txMapC(\varX)$ for some $\varX$ such that
  $\actGen{\varX}{\valVi}$ is fired in $\runS$, then:
  \begin{enumerate}
  \item if $(\txT,i) \not\in \ran(\txMapC)$ then $(\txT,i)$ is unspendable;
  \item if $(\txT,i) = \txMapC(\varY)$ for some $\varY$, then $(\txT,i)$ is spendable, 
    and $\confDep[\varY]{\pmvA}{\valV:\tokT}$ occurs in $\confS{\runS}$,
    where $\pmvA = (\txT,i).\tkown$.
  \end{enumerate}
\end{lem}

Note that all the results above require as hypothesis that \mbox{$\coherRel{\runS}{\runC}{}$}.
In the rest of this~\namecref{sec:computational-soundness} we show that,
with overwhelming probability, for each computational run $\runC$ there exists
a symbolic run $\runS$ which is coherent with $\runC$.
Actually, our computational soundness result 
is more precise than that.
First, we consider only a subset $\PartT$ of users to be honest, 
while we consider all the others dishonest: 
without loss of generality, we model them as a single adversary $\Adv$.
We assume that both honest users and the adversary have a \emph{strategy},
which allows them to decide the next actions, according to the past run.
Our computational soundness result establishes that,
with overwhelming probability, 
any computational run conforming to the (computational) strategies 
has a corresponding symbolic run conforming to the corresponding (symbolic) strategies.
Consequently, if there exists an attack at the computational level,
then the attack is also observable at the symbolic level.

\paragraph*{Symbolic strategies}

The strategy $\stratS{\pmvA}$ of a honest user $\pmvA$ 
is a PTIME algorithm which allows $\pmvA$ to select which action(s) to perform, 
among those permitted by the token semantics.
$\stratS{\pmvA}$ receives as input a finite symbolic run $\runS$,
and outputs a finite set of enabled actions,
with the constraint that $\stratS{\pmvA}$ cannot output authorizations for $\pmvB \neq \pmvA$.
We require strategies to be \emph{persistent}: 
if on a run $\stratS{\pmvA}$ chooses an action $\labS$, 
and $\labS$ is not taken as the next step in the run 
(\eg, because some other user acts earlier), 
then $\stratS{\pmvA}$ must still choose $\labS$ after that step, 
if still enabled.
%
The \emph{adversary} $\Adv$ acts on behalf of all the dishonest users,
and controls the scheduling among all users (including the honest ones).
Her symbolic strategy $\stratS{\Adv}$ is a PTIME algorithm 
taking as input the current run and the sets of moves outputted by the strategies of honest users. 
The output of $\stratS{\Adv}$ 
is a single action $\labS$ (to be appended to the current run).
To rule out authorization forgeries,
we require that if $\labS$ is an authorization by some honest $\pmvA$, 
then it must be chosen by $\stratS{\pmvA}$.
%
%
Fixing a set of strategies $\stratSSet$ 
--- both for the honest users and for the adversary --- 
we obtain a unique run, which is made by the sequence of actions 
chosen by $\stratS{\Adv}$
when taking as input the outputs of the honest users' strategies.
We say that this run is \emph{conformant} to $\stratSSet$.
When $\stratS{\Adv} \not\in \stratSSet$,
we say that $\runS$ conforms to $\stratSSet$
when there exists some $\stratS{\Adv}$ such that
$\runS$ conforms to $\stratSSet \cup \setenum{\stratS{\Adv}}$.

\paragraph*{Computational strategies}

A computational strategy $\stratC{\pmvA}$ for a honest user $\pmvA$ is a PPTIME algorithm
which receives as input a computational run $\runC$,
and outputs a finite set of computational labels.
The choice among these labels is controlled by $\Adv$'s strategy, specified below.
We assume that each user knows her private key, and the public keys of all the other users.
We impose a few sanity constraints: 
\begin{inlinelist}
\item we forbid $\pmvA$ to impersonate another user;
\item \label{item:stratC:witness} 
  if the strategy outputs a transaction $\txT$, 
  then $\txT$ must be a valid update of the blockchain $\bcB[\runC]$
  obtained from the run in input, and
  all the witnesses of $\txT$ have already been broadcast in the run;
\item finally, strategies must be persistent, similarly to the symbolic case.
\end{inlinelist}
%
$\Adv$'s computational strategy $\stratC{\Adv}$
is a PPTIME algorithm taking as input a run $\runC$
and the moves chosen by each honest user.
The strategy gives as output a single computational label, 
to be appended to the run.
We assume that $\Adv$ can impersonate any other user.
Given a symbolic strategy $\stratS{\pmvA}$, we can obtain a computational strategy
$\stratC{\pmvA} = \stratMap(\stratS{\pmvA})$
by implementing the symbolic actions as the corresponding computational ones
(this can be done using the same technique as in~\cite{BZ18bitml}).
%
%
Given a set of computational strategies $\stratCSet$ 
--- both for the honest users and for the adversary --- 
we probabilistically generate a \emph{conformant} computational run, 
made by the sequence of actions chosen by $\stratC{\Adv}$
when taking as input the outputs of the honest users' strategies.

\paragraph*{Computational soundness}

We are now ready to establish our main result.
We assume that cryptographic primitives are secure,
\ie, hashes are collision resistant and signatures cannot be forged
(except with negligible probability).

\begin{thm}
  \label{th:computational-soundness}
  Let $\stratSSet$ be a set of symbolic strategies 
  for all $\pmvA \in \PartT$.
  Let $\stratCSet$ be a set of computational strategies
  such that $\stratC{\pmvA} = \stratMap(\stratS{\pmvA})$
  for all $\pmvA \in \PartT$,
  and including an arbitrary adversary strategy $\stratC{\Adv}$.
  For each $k \in \Nat$ and security parameter $\eta$, we define the following experiment:
  \begin{enumerate}
  \item generate $\runC$ conforming to $\stratCSet$, with $|\runC| \leq \eta^k$;
  \item if there exists $\runS$ conforming to $\stratSSet$ such that $\coherRel{\runS}{\runC}{}$
    then return 1, otherwise return 0.
  \end{enumerate}
  Then, the experiment returns 1 with overwhelming probability \wrt the security parameter $\eta$.
\end{thm}

As an application of our results, we lift the balance preservation result
of \Cref{lem:tokens-s:tokval-preservation} from the symbolic to the computational model.
We start by defining the balance of a computational token.
Then, in~\Cref{th:tokvals-tokvalc} we establish the equivalence between 
symbolic token balance and the computational one.

\begin{defn}[Balance of a computational token]
  \label{def:token-c:tokval}
  For all computational runs $\runC$ and transaction outputs $\txout$, let:
  \[
  P_{\txout} \; = \; \Big\{ \txouti \; \Big|
  \begin{array}{l}
    \txouti \in \utxo{{\bcB[\runC]}} \;\text{and}\; \txouti \text{ is spendable} \; \text{and} \\
    \txouti.\tkid = \txout \;\text{and}\; \txouti.\txscript = \tokScript
  \end{array}
  \Big\}
  \]
  We define the balance of the computational token $\txout$ in $\runC$ as:
  \[
  \tokvalC{\txout}{\runC} 
  \; = \;
  \sum_{\txouti \in P_{\txout}} \txouti.\tkval
  \]
\end{defn}

\begin{thm}
  \label{th:tokvals-tokvalc}
  Let $\coherRel{\runS}{\runC}{}$.
  If $\actGen{\varX}{\valV}$ is fired in $\runS$, generating a fresh token $\tokT$,
  then:
  \[
  \tokval{\tokT}{\runS} \; = \; \tokvalC{\txMapC(\varX)}{\runC} 
  \]
\end{thm}
\begin{proof}
  We first prove that $\tokval{\tokT}{\runS} \leq \tokvalC{\txMapC(\varX)}{\runC}$.
  For each $\confDep[\varY]{\pmvA}{\valV:\tokT}$ in $\confS{\runS}$,
  by~\Cref{lem:coher:s-to-c} there exists a distinct spendable
  $(\txT,i) = \txMapC(\varY)$ in $\utxo{\bcB[\runC]}$
  such that $(\txT,i).\tkown = \pmvA$, 
  $(\txT,i).\txscript = \tokScript$, 
  $(\txT,i).\tkval = \valV$, 
  and $(\txT,i).\tkid = \txMapC(\varX)$.
  Then, $(\txT,i) \in P_{\txMapC(\varX)}$,
  and so $\valV$ contributes to $\tokvalC{\txMapC(\varX)}{\runC}$.
  We now prove that $\tokval{\tokT}{\runS} \geq \tokvalC{\txMapC(\varX)}{\runC}$.
  Let $\txout \in P_{\txMapC(\varX)}$.
  Since $\txout \in \utxo{\bcB[\runC]}$ and it is spendable,
  by item 2 of~\Cref{lem:coher:c-to-s} some deposit
  $\confDep[]{\pmvA}{\valV:\tokT}$ occurs in $\confS{\runS}$, with $\pmvA = \txout.\tkown$.
  Then, $\valV$ contributes to $\tokval{\tokT}{\runC}$.
\end{proof}

\section{Related work and conclusions}
\label{sec:conclusions}

\newcommand{\ZCB}{\ensuremath{\contrFmt{\it ZCB}}\xspace}

We have proposed a secure and efficient implementation of fungible tokens on Bitcoin,
exploiting neighbourhood covenants, a powerful yet simple extension of the Bitcoin script language.
The security of our construction is established as a computational soundness result
(\Cref{th:computational-soundness}).
This guarantees that adversaries at the Bitcoin level cannot make tokens diverge from
the ideal behaviour specified by the symbolic token semantics, 
unless with negligible probability.

To keep the presentation simple, we have slightly limited the functionality of tokens,
making split/join/exchange actions operate on just two deposits, and omitting time constraints.
Lifting these restrictions would only affect the size, but not the complexity, of our technical development.
Further, it would allow to use tokens \emph{as is} within high-level languages for Bitcoin contracts, 
\eg BitML~\cite{BZ18bitml,bitmlracket},
simplifying the design of financial contracts which manage tokens.
For instance, we could model as follows a basic zero-coupon bond~\cite{PeytonJones00icfp}
where an investor $\pmvA$ pays upfront to a bank $\pmvB$ $5$ units of a token $\tokT$,
and receives back $1 \BTC$ after a maturity date $t$:
\[
  \splitname ( 
  \splitB{5\mbox{:}\tokT}{\withdrawC{\pmvB}}
  \mid
  \splitB{1\mbox{:}\BTC}{\afterC{t}{\withdrawC{\pmvA}}}
  )
\]

A research question arising from our work
is how to exploit neighbourhood covenants in the design of high-level languages for Bitcoin contracts.
Besides enhancing the expressiveness of these languages 
(\eg, by allowing for unbounded recursion), 
neighbourhood covenants would enable a simpler compilation technique,
compared \eg that used in BitML.
Indeed, to guarantee the liveness of contracts, BitML requires the participants in a contract to pre-exchange their signatures on all the transactions obtained by the compiler.
By using covenants, we could avoid this overhead, since the logic which controls that a contract action is permitted is not based on signatures, but is implemented by the covenant.

  \paragraph*{Related work}
\label{sec:related}


In Ethereum, similarly to our approach, tokens are implemented
on top of the platform using a smart contract,
following the ERC-20 and ERC-712 standards~\cite{ERC20,ERC721}.
An alternative approach is to provide a native support for tokens:
the works~\cite{Chakravarty20utxoma,Chakravarty20tokens} follow this
approach, proposing an extension of the UTXO model used in the Cardano blockchain~\cite{Chakravarty20wtsc}.
The work~\cite{Zahnentferner20iacr} also proposes an extension of the UTXO model
that supports native tokens.
Unlike the previous approaches, in this model there is no privileged crypto-currency:
transaction fees are paid in the same currency which is being exchanged.
Algorand~\cite{algorandAssets} supports native tokens in an account-based model,
allowing their minting, burn, transfer, 
and their exchange through atomic groups of transactions.


The first proposals of covenants in Bitcoin date back at least to 2013~\cite{bitcointalk-covenants}. 
Nevertheless, their inclusion into the official Bitcoin protocol is still uncertain,
mainly because of the cautious approach of the Bitcoin community
to accept extensions~\cite{BIP0002}.
The emerging of Bitcoin layer-2 protocols like the Lightning Network~\cite{Poon15lighting}
has revived the interest in covenants,
as witnessed by a recent Bitcoin Improvement Proposal (BIP 119~\cite{BIP119}).
Currently, covenants are supported by Bitcoin Cash~\cite{cashscript-covenants},
a mainstream blockchain platform originated from a fork of Bitcoin.

The work~\cite{Moser16bw} proposes a new opcode {\small\sf CheckOutputVerify}
to explicitly constrain the outputs of the redeeming transaction,
while \cite{Oconnor17bw} implements covenants 
by extending the current implementation of {\small\sf versig}
with a new opcode {\small\sf CheckSigFromStack},
which can check a signature on \emph{arbitrary} data on the stack.
Both~\cite{Moser16bw} and \cite{Oconnor17bw} can constrain
the script of the redeeming transaction
to contain the same covenant of the spent one, 
enabling recursive covenants similarly to our $\verrec{\rtxo{}{n}}$.

An alternative approach is to implement covenants without adding new opcodes.
This approach is proposed by~\cite{Swambo20bitcoin},
which relies on modified signature verification scheme,
allowing users to sign a transaction \emph{template}, 
\ie to sign only parts of a transaction, leaving the state parameters variable.
The idea that covenants would allow to implement state machines on Bitcoin was 
first made by~\cite{Oconnor17bw}:
in~\Cref{th:covenants:turing-completeness} we show that, 
using our neighbourhood covenants, Bitcoin can actually simulate counter machines, 
assuming unbounded integers.

The work~\cite{Chepurnoy19cbt} describes an approach to extend the UTXO model 
with a variant of covenants which, similarly to ours, can access the redeeming 
and the sibling transactions.
Besides this, the extension in \cite{Chepurnoy19cbt} features registers, 
loops, complex data structures, and native tokens,
making it quite distant from an implementation on Bitcoin.
Compared to~\cite{Chepurnoy19cbt}, our extension devises a minimal extension to the UTXO model, 
so to allow for an efficient implementation on Bitcoin, and support secure fungible tokens.

The work~\cite{BLZ20isola} introduces a formal model of covenants,
which can be implemented in Bitcoin with modifications similar to the ones discussed
in~\Cref{sec:covenants}.
However, since the covenants in~\cite{BLZ20isola} only constrain the redeeming transaction,
their implementation is simpler, since it does not require to extend the UTXO set with the parents. 
As a downside, the covenants of~\cite{BLZ20isola} do not allow to implement fungible tokens.

  \paragraph*{Acknowledgements}
Massimo Bartoletti is partially supported by
Convenzione Fondazione di Sardegna e Atenei Sardi project F74I19000900007 \emph{``ADAM''}.
Stefano Lande is supported by P.O.R.\ F.S.E.\ 2014-2020.
Roberto Zunino is partially supported by
MIUR PON \textit{``Distributed Ledgers for Secure Open Communities''}.


\begin{figure*}[tp!]
  \centering
  \begin{center}
\begin{lstlisting}[language=balzac,numbers=left,numbersep=7pt,classoffset=1,morekeywords={},classoffset=2,morekeywords={},framexbottommargin=0pt,framexleftmargin=0pt,framexrightmargin=-80pt,xleftmargin=80pt]
if not verrec(ptxo(1)) then                        // ctxo is a token generator
      ctxo.op = 0                                                 // gen action
  and ctxo.tkid = txid(ptxo(1))                                     // token id
  and ptxo(1).val = 0                                    // spent txo has 0 BTC
  and outlen(ctxo) = 1                                         // gen has 1 out
  and ctxo.tkval > 0                                    // positive token value
else true

and

if rtxo.op = 1 then            /******************** BURN ********************/ 
      versig(ctxo.owner,rtx.wit)                                 // check owner
  and verscr(false,rtxo(1))                             // make rtx unspendable
  and inlen(rtxo(1)) = 1                                        // rtx has 1 in
  and outlen(rtxo(1)) = 1                                      // rtx has 1 out

else if rtxo.op = 2 then      /******************** SPLIT ********************/
      versig(ctxo.owner,rtx.wit)                                 // check owner
  and verrec(rtxo(1))                                       // covenants on rtx
  and verrec(rtxo(2))  
  and inlen(rtxo(1)) = 1                                        // rtx has 1 in
  and outlen(rtxo(1)) = 2                                     // rtx has 2 outs
  and rtxo(1).tkval >= 0                                // positive token value
  and rtxo(2).tkval >= 0
  and rtxo(1).owner = ctxo.owner                              // preserve owner  
  and rtxo(1).tkid = ctxo.tkid                                 // preserve tkid
  and rtxo(2).tkid = ctxo.tkid  
  and rtxo(1).tkval + rtxo(2).tkval = ctxo.tkval              // preserve value

else if rtxo.op = 3 then       /******************** JOIN ********************/
      inlen(rtxo(1)) = 2                                        // rtx has 2 in
  and outlen(rtxo(1)) = 1                                      // rtx has 1 out
  and verrec(rtxo(1))                                       // coventant on rtx
  and verrec(stxo(2))                               // covenants on both inputs
  and verrec(stxo(1))        
  and ctxo.tkid = rtxo(1).tkid                             // preserve token id
  and versig(ctxo.owner, rtx.wit)                                // check owner
  and rtxo(1).tkval = stxo(1).tkval + stxo(2).tkval           // preserve value

else if rtxo.op = 4 then       /******************** XCHG ********************/
      inlen(rtxo(1)) = 2                                       // rtx has 2 ins
  and outlen(rtxo(1)) = 2                                     // rtx has 2 outs
  and versig(ctxo.owner, rtx.wit)                                // check owner
  and verrec(stxo(1))                                    // covenant on input 1
  and verrec(rtxo(1))                                     // covenant on rtx(1)
  and rtxo(1).owner = stxo(2).owner                              // swap owners
  and rtxo(2).owner = stxo(1).owner  
  and rtxo(1).tkval = stxo(1).tkval                           // preserve value 
  and rtxo(1).tkid = stxo(1).tkid                              // preserve tkid 
  if verrec(stxo(2)) then                    /***** EXCHANGE TOKEN/TOKEN *****/
        verrec(rtxo(2))                                   // covenant on rtx(2)
    and rtxo(2).tkval = stxo(2).tkval                         // preserve tkval
    and rtxo(2).tkid = stxo(2).tkid                            // preserve tkid
  else                                         /***** EXCHANGE TOKEN/BTC *****/
        verscr(versig(ctxo.owner, rtx.wit), rtxo(2))
    and rtxo(2).val = stxo(2).val                               // preserve BTC

else if rtxo.op = 5 then       /******************** GIVE ********************/
      inlen(rtxo(1)) = 1                                        // rtx has 1 in
  and outlen(rtxo(1)) = 1                                      // rtx has 1 out
  and versig(ctxo.owner, rtx.wit)                                // check owner
  and verrec(rtxo(1))                                     // covenant on rtx(1)
  and rtxo(1).tkid = ctxo.tkid                                 // preserve tkid
  and rtxo(1).tkval = ctxo.tkval                              // preserve value

else false
\end{lstlisting}
  \end{center}
  \vspace{-10pt}
  \caption{Full token script.}
  \label{fig:full-token:script}
\end{figure*}

\bibliographystyle{IEEEtran}
\bibliography{main}

\iftoggle{arxiv}{
\clearpage
\appendices
\iftoggle{arxiv}
{\section{Supplementary material for~\Cref{sec:tokens:c}}}
{}
\label{app:coherence}

We now formalize the coherence relation between symbolic and computational runs 
sketched in~\Cref{sec:computational-soundness}.
Below, we call $\tokScript$ the token script in~\Cref{fig:full-token:script},
and we call $\btcScript$ the script:
\[
  \versig{\ctxo{\tkown}}{\rtx.\txWit{}}
\]

We inductively define the relation: 
\[
\coher{\runS}{\runC}{\txMapC}{\tkMapC}{\burnMapC}
\]
where:
\begin{itemize}
\item $\runS$ is a symbolic run;
\item $\runC$ is a computational run;
\item $\txMapC$ is an injective function, mapping the names $\varX$ occurring in $\confS{\runS}$
  (except those only in authorizations)
  to unspent transaction outputs $(\txT,i)$ (where $\txT$ occurs in $\runC$);
\item $\tkMapC$ is an injective function, mapping the tokens $\tokT$ occurring in $\confS{\runS}$
  to unspent transaction outputs $(\txT,i)$ (where $\txT$ occurs in $\runC$).
\item $\burnMapC$ is an injective function, mapping the names  $\varX$ occurring in $\confS{\runS}$ only in
  $\burnOp$ authorizations to off-chain transactions.
\end{itemize}
We require the relation to respect the following invariants:
\begin{itemize}

\item if $\txMapC(x) = (\txT,i)$, then for all $(\txTi,j) \in \txT.\txIn{}$,
  there exists $y$ such that $\txMapC(y) = (\txTi,j)$.

\item if $\confS{\runS} = \confDep[\varX]{\pmvA}{\valV:\BTC} \mid \cdots$,
then there exist $\txT,i$ such that $\txMapC(\varX) = (\txT,i)$ and:
\begin{enumerate}
\item $(\txT,i).\txarg = \constPK[\pmvA]$
\item $(\txT,i).\txscript = \versig{\ctxo{}.\txarg}{\rtx.\txWit{}}$
\item $(\txT,i).\txval = \valV$.
\end{enumerate}

\item if $\confS{\runS} = \confDep[\varX]{\pmvA}{\valV:\tokT} \mid \cdots$
  with $\tokT \neq \BTC$
  then there exist $\txT,i$, $\txTi,j$ such that 
  $\txMapC(\varX) = (\txT,i)$, 
  $\tkMapC(\tokT) = (\txTi,j)$, and:
  \begin{enumerate}
  \item $(\txT,i).\tkop \in 0..5$
  \item $(\txT,i).\tkid = (\txTi,j)$
  \item $(\txT,i).\tkown = \pmvA$
  \item $(\txT,i).\tkval = \valV$
  \item $(\txT,i).\txscript = \tokScript$
  \item $(\txT,i).\txval = 0$
  \end{enumerate}
  
\end{itemize}

\paragraph*{Base case}
$\coher{\runS}{\runC}{\txMapC}{\tkMapC}{\burnMapC}$
holds if the following conditions hold:
\begin{itemize}
\item $\runS$ and $\runC$ are initial, \ie: $\runS$ is \emph{initial} if it contains only bitcoin deposits,
and $\runC$ only contains from a coinbase transaction $\txT[0]$;
\item $\txMapC{}$ injectively maps exactly the names $\varX$ of the deposits $\confDep[\varX]{\pmvA}{\valV:\BTC}$
  in $\confS{\runS}$ to outputs $(\txT[0],i)$ of the form:
  \begin{enumerate}
  \item $(\txT[0],i).\txarg = \constPK[\pmvA]$
  \item $(\txT[0],i).\txscript = \versig{\ctxo{}.\txarg}{\rtx.\txWit{}}$
  \item $(\txT[0],i).\txval = \valV$
  \end{enumerate}
\item $\dom(\tkMapC) = \emptyset$
\end{itemize}

\paragraph*{Inductive case 1}
$\coher{\runS\labS}{\runC\labC}{\txMapC}{\tkMapC}{\burnMapC}$ holds if
$\coher{\runS}{\runC}{\txMapCi}{\tkMapCi}{\burnMapCi}$ and one of the following cases applies:
\begin{enumerate}[(1)]

\item \label{item:coherence:actGen}
$\labS = \actGen{\varX}{\valV}$. 
In $\confS{\runS}$ there exists a deposit $\confDep[\varX]{\pmvA}{0:\BTC}$
and $\labC = \txT$, where:
\begin{itemize}
\item $\txT.\txIn{} = \txMapCi(\varX)$
\item $\txT$ has exactly one output
\item $(\txT,1).\tkop = 0$
\item $(\txT,1).\tkid = \txMapCi(\varX)$
\item $(\txT,1).\tkown = \pmvA$
\item $(\txT,1).\tkval = \valV$
\item $(\txT,1).\txscript = \tokScript$
\item $(\txT,1).\txval = 0$
\end{itemize}
In $\confS{\runS\labS}$ there exists a deposit $\confDep[\varY]{\pmvA}{\valV:\tokT}$ 
which did not exist in $\confS{\runS}$.
The mapping $\txMapC$ inherits all the bindings of $\txMapCi$ except that for $\varX$,
and extends it with the binding $\varY \mapsto (\txT,1)$.
The mapping $\tkMapC$ extends $\tkMapCi$ with the binding $\tokT \mapsto \txMapCi(\varX)$.
The mapping $\burnMapC$ is equal to $\burnMapCi$.

\item \label{item:coherence:actSplit}
$\labS = \actSplit{\varX}{\valV}{\pmvB}$.
In $\confS{\runS}$ there exists a deposit $\confDep[\varX]{\pmvA}{(\valV+\valVi):\tokT}$
and $\labC = \txT$.
The following conditions hold:
\begin{itemize}
\item $\txT.\txIn{} = \txMapCi(\varX)$
\item $\txT$ has exactly two outputs (of indices 1 and 2)
\item $(\txT,1).\tkown = \pmvA$
\item $(\txT,2).\tkown = \pmvB$
\end{itemize}
Further, if $\tokT = \BTC$, the following conditions hold:
\begin{itemize}
  \item $(\txT,1).\txscript = (\txT,2).\txscript = \btcScript$
  \item $(\txT,1).\txval = \valV$
  \item $(\txT,2).\txval = \valVi$
\end{itemize}
Further, if $\tokT \neq \BTC$, the following conditions hold:
\begin{itemize}
  \item $(\txT,1).\tkop = 2$
  \item $(\txT,1).\tkid = (\txT,2).\tkid = \txMapCi(\varX).\tkid$
  \item $(\txT,1).\tkval = \valV$
  \item $(\txT,2).\tkval = \valVi$
  \item $(\txT,1).\txscript = (\txT,2).\txscript = \tokScript$    
  \item $(\txT,1).\txval = (\txT,2).\txval = 0$
\end{itemize}
In both cases, in $\confS{\runS\labS}$ there exist two deposits 
$\confDep[\varY]{\pmvA}{\valV:\tokT}$ and $\confDep[\varYi]{\pmvB}{\valVi:\tokT}$
which did not exist in $\confS{\runS}$.
The mapping $\txMapC$ inherits all the bindings of $\txMapCi$ except that for $\varX$,
and extends it with the bindings $\varY \mapsto (\txT,1)$ and $\varYi \mapsto (\txT,2)$.
The mapping $\tkMapC$ is equal to $\tkMapCi$.
The mapping $\burnMapC$ is equal to $\burnMapCi$.

\item \label{item:coherence:actJoin}
$\labS = \actJoin{\varX}{\varY}{\pmvC}$.
In $\confS{\runS}$ there exist two deposits 
$\confDep[\varX]{\pmvA}{\valV:\tokT}$ and $\confDep[\varY]{\pmvB}{\valVi:\tokT}$,
and $\labC = \txT$.
The following conditions hold:
\begin{itemize}
\item $\txT$ has exactly two inputs
\item $\txT.\txIn[1]{} = \txMapCi(\varX)$
\item $\txT.\txIn[2]{} = \txMapCi(\varY)$
\item $\txT$ has exactly one output
\item $(\txT,1).\tkown = \pmvC$
\end{itemize}
Further, if $\tokT = \BTC$, the following conditions hold:
\begin{itemize}
  \item $(\txT,1).\txval = \valV+\valVi$
  \item $(\txT,1).\txscript = \btcScript$
\end{itemize}
Further, if $\tokT \neq \BTC$, the following conditions hold:
\begin{itemize}
  \item $(\txT,1).\tkop = 3$
  \item $(\txT,1).\tkid = \txMapCi(\varX).\tkid = \txMapCi(\varY).\tkid$
  \item $(\txT,1).\tkval = \valV+\valVi$
  \item $(\txT,1).\txscript = \tokScript$    
  \item $(\txT,1).\txval = 0$
\end{itemize}
In both cases, in $\confS{\runS\labS}$ there exist one deposit
$\confDep[\varZ]{\pmvA}{(\valV+\valVi):\tokT}$
which did not exist in $\confS{\runS}$.
The mapping $\txMapC$ inherits all the bindings of $\txMapCi$ except those for $\varX$ and $\varY$,
and extends it with the binding $\varZ \mapsto (\txT,1)$.
The mapping $\tkMapC$ is equal to $\tkMapCi$.
The mapping $\burnMapC$ is equal to $\burnMapCi$.

\item \label{item:coherence:actExchange}
$\labS = \actExchange{\varX}{\varY}$.
In $\confS{\runS}$ there exist two deposits 
$\confDep[\varX]{\pmvA}{\valV:\tokT}$ and $\confDep[\varY]{\pmvB}{\valVi:\tokTi}$,
$\tokT \neq \BTC$, 
and $\labC = \txT$.
The following conditions hold:
\begin{itemize}
\item $\txT$ has exactly two inputs and two outputs
\item $\txT.\txIn[1]{} = \txMapCi(\varX)$
\item $\txT.\txIn[2]{} = \txMapCi(\varY)$
\item $(\txT,1).\tkop = 4$
\item $(\txT,1).\tkown = \pmvB$
\item $(\txT,2).\tkown = \pmvA$
\item $(\txT,1).\tkid = \txMapCi(\varX).\tkid$
\item $(\txT,1).\tkval = \valV$
\item $(\txT,1).\txval = 0$
\item $(\txT,1).\txscript = \tokScript$
\end{itemize}
Further, if $\tokTi = \BTC$, the following conditions hold:
\begin{itemize}
  \item $(\txT,2).\txval = \valVi$
  \item $(\txT,2).\txscript = \btcScript$
\end{itemize}
Further, if $\tokTi \neq \BTC$, the following conditions hold:
\begin{itemize}
  \item $(\txT,2).\txscript = \tokScript$
  \item $(\txT,2).\tkval = \valVi$
  \item $(\txT,2).\tkid = \txMapCi(\varY).\tkid$
  \item $(\txT,2).\txval = 0$
\end{itemize}
In both cases, in $\confS{\runS\labS}$ there exist two deposits
$\confDep[\varXi]{\pmvA}{\valVi:\tokTi}$ and $\confDep[\varYi]{\pmvB}{\valV:\tokT}$
which did not exist in $\confS{\runS}$.
The mapping $\txMapC$ inherits all the bindings of $\txMapCi$ except those for $\varX$ and $\varY$,
and extends it with the bindings $\varXi \mapsto (\txT,2)$ and $\varYi \mapsto (\txT,1)$.
The mapping $\tkMapC$ is equal to $\tkMapCi$.
The mapping $\burnMapC$ is equal to $\burnMapCi$.

\item \label{item:coherence:actGive} 
$\labS = \actGive{\varX}{\pmvB}$.
In $\confS{\runS}$ there exists a deposit $\confDep[\varX]{\pmvA}{\valV:\tokT}$
and $\labC = \txT$.
One of the following subcases must apply:
\begin{enumerate}
\item \label{item:coherence:actGive:a}
  The following conditions hold:
  \begin{itemize}
  \item $\txT.\txIn{} = \txMapCi(\varX)$
  \item $\txT$ has exactly one output
  \item $(\txT,1).\tkown = \pmvB$
  \end{itemize}
  Further, if $\tokT = \BTC$, the following conditions hold:
  \begin{itemize}
  \item $(\txT,1).\txscript = \btcScript$
  \item $(\txT,1).\txval = \valV$
  \end{itemize}
  Further, if $\tokT \neq \BTC$, the following conditions hold:
  \begin{itemize}
  \item $(\txT,1).\tkop = 5$ 
  \item $(\txT,1).\tkid = \txMapCi(\varX).\tkid$
  \item $(\txT,1).\tkval = \valV$
  \item $(\txT,1).\txscript = \tokScript$    
  \item $(\txT,1).\txval = 0$
  \end{itemize}
  In both cases, in $\confS{\runS\labS}$ there exists a deposit
  $\confDep[\varY]{\pmvB}{\valV:\tokT}$ which did not exist in $\confS{\runS}$.
  
\item \label{item:coherence:actGive:b}
  The following conditions hold:
  \begin{itemize}
  \item $\txT$ has exactly two inputs and two outputs 
  \item $(\txT,1).\txIn{} = \txMapCi(\varX)$
  \item $(\txT,2).\txIn{} \not\in \ran(\txMapCi)$
  \item $(\txT,1).\tkop = 4$ 
  \item $(\txT,1).\txscript = \tokScript$
  \item $(\txT,1).\tkid = \txMapCi(\varX).\tkid$
  \item $(\txT,1).\tkown = \txT.\txIn[2]{}.\tkown$ 
  \item $(\txT,1).\tkval = \valV$
  \item $(\txT,1).\txval = 0$
  \end{itemize}
  Note: in this case we are using the symbolic $\giveOp$ action to match the appending of 
  a transaction $\txT$ which consumes two inputs: 
  the first one is mapped by $\txMapCi$ and encodes a user-defined token $\tokT$,
  while the second one is not mapped by $\txMapCi$.
  So, the computational action encoded by $\txT$ is an $\exchangeOp$ between $\tokT$ and some $\tokTi$.
  Appending $\txT$ actually preserves the token identifier and its value, 
  just changing its owner. 

\item \label{item:coherence:actGive:c}
  The following conditions hold:
  \begin{itemize}
  \item $\txT$ has exactly two inputs and two outputs 
  \item $(\txT,1).\txIn{} \not\in \ran(\txMapCi)$
  \item $(\txT,2).\txIn{} = \txMapCi(\varX)$
  \item $(\txT,1).\tkop = 4$ 
  \item $(\txT,2).\txscript = \tokScript$
  \item $(\txT,2).\tkid = \txMapCi(\varX).\tkid$
  \item $(\txT,2).\tkown = \txT.\txIn[1]{}.\tkown$ 
  \item $(\txT,2).\tkval = \valV$
  \item $(\txT,2).\txval = 0$
  \end{itemize}

\end{enumerate}
In all the subcases above,
the mapping $\txMapC$ inherits all the bindings of $\txMapCi$ except that for $\varX$,
and extends it with the binding 
$\varY \mapsto (\txT,1)$ (in subcases \ref{item:coherence:actGive:a}-\ref{item:coherence:actGive:b}),
or with $\varY \mapsto (\txT,2)$ (in subcase \ref{item:coherence:actGive:c}).
The mapping $\tkMapC$ is equal to $\tkMapCi$.
The mapping $\burnMapC$ is equal to $\burnMapCi$.

\item \label{item:coherence:actBurn} 
$\labS = \actBurn{\vec{\varX}}{\varY}$.
Let $\vec{\varX} = \varX[1] \cdots \varX[n]$.
In $\confS{\runS}$ there exist $n$ deposits $\confDep[{\varX[i]}]{\pmvA[i]}{\valV[i]:\tokT[i]}$,
and $\labC = \txT = \burnMapCi(\varY)$.
One of the following subcases must apply:
\begin{enumerate} 

\item $\tokT[1] \neq \BTC$ and $n=1$. Further:
  \begin{itemize}
  \item $\txT.\txIn{} = \txMapCi(\varX[1])$
  \item $\txT$ has exactly one output
  \item $(\txT,1).\tkop = 1$
  \item $(\txT,1).\txscript = \false$
  \item $(\txT,1).\txval = 0$
  \end{itemize}
  Note: in this case we are using the symbolic $\burnOp$ action to match the appending of 
  a transaction $\txT$ which consumes one input mapped by $\txMapCi$,
  and encodes a user-defined token.
  The script of $\txT$ is set to $\false$, actually destroying the token. 

\item $\tokT[1] = \BTC$ and $\tokT[i] = \BTC$ for all $i \in 1..n$. Further:
  \begin{itemize}
  \item $\txT.\txIn{}$ contains at least $\txMapCi(\varX[i])$ for all $i \in 1..n$,
    but it does not contain any other output in $\ran(\txMapCi)$:
    \[
    \setcomp{\varX}{\txMapCi(\varX) \in \txT.\txIn{}} = \setenum{\varX[1], \ldots, \varX[n]}
    \]
  \item $\txT$ does not correspond to any action $\labS$ discussed in the 
    items~\ref{item:coherence:actGen}-\ref{item:coherence:actGive} above.
  \end{itemize}

\end{enumerate}
The mapping $\txMapC$ is equal to $\txMapCi$.
The mapping $\tkMapC$ is equal to $\tkMapCi$.
The mapping $\burnMapC$ inherits all the bindings in $\burnMapCi$, except that for $\varY$.

\item \label{item:coherence:authGen}
$\labS = \labAuth[\varX]{\pmvA}{\actGen{\varX}{\valV}}$.
In $\confS{\runS}$ there exists a deposit $\confDep[\varX]{\pmvA}{0:\BTC}$.
$\labC = \pmvB \rightarrow *:m$, 
where $m$ is the signature of $\pmvA$ on a transaction $\txT$ satisfying the 
conditions in item~\ref{item:coherence:actGen}.
The mappings $\txMapC$, $\tkMapC$, $\burnMapC$ are equal to
$\txMapCi$, $\tkMapCi$, $\burnMapCi$, respectively.

\item \label{item:coherence:authSplit}
$\labS = \labAuth[\varX]{\pmvA}{\actSplit{\varX}{\valV}{\pmvB}}$.
In $\confS{\runS}$ there exists a deposit $\confDep[\varX]{\pmvA}{(\valV+\valVi):\tokT}$.
$\labC = \pmvB \rightarrow *:m$, 
where $m$ is the signature of $\pmvA$ on a transaction $\txT$ satisfying the 
conditions in item~\ref{item:coherence:actSplit}.
The mappings $\txMapC$, $\tkMapC$, $\burnMapC$ are equal to
$\txMapCi$, $\tkMapCi$, $\burnMapCi$, respectively.

\item \label{item:coherence:authJoin}
$\labS = \labAuth[\varZ]{\pmvP}{\actJoin{\varX}{\varY}{\pmvC}}$,
where $(\pmvP,\varZ) \in \setenum{(\pmvA,\varX),(\pmvB,\varY)}$,
for some $\pmvA$ and $\pmvB$.
In $\confS{\runS}$ there exist two deposits 
$\confDep[\varX]{\pmvA}{\valV:\tokT}$ and $\confDep[\varY]{\pmvB}{\valVi:\tokT}$.
$\labC = \pmvB \rightarrow *:m$, 
where $m$ is the signature of $\pmvP$ on a transaction $\txT$ satisfying the 
conditions in item~\ref{item:coherence:actJoin}.
The mappings $\txMapC$, $\tkMapC$, $\burnMapC$ are equal to
$\txMapCi$, $\tkMapCi$, $\burnMapCi$, respectively.

\item \label{item:coherence:authExchange}
$\labS = \labAuth[\varZ]{\pmvP}{\actExchange{\varX}{\varY}}$,
where $(\pmvP,\varZ) \in \setenum{(\pmvA,\varX),(\pmvB,\varY)}$,
for some $\pmvA$ and $\pmvB$.
In $\confS{\runS}$ there exist two deposits 
$\confDep[\varX]{\pmvA}{\valV:\tokT}$ and $\confDep[\varY]{\pmvB}{\valVi:\tokTi}$, with $\tokT \neq \BTC$.
$\labC = \pmvB \rightarrow *:m$, 
where $m$ is the signature of $\pmvP$ on a transaction $\txT$ satisfying the 
conditions in item~\ref{item:coherence:actExchange}.
The mappings $\txMapC$, $\tkMapC$, $\burnMapC$ are equal to
$\txMapCi$, $\tkMapCi$, $\burnMapCi$, respectively.

\item \label{item:coherence:authGive}
$\labS = \labAuth[\varX]{\pmvA}{\actGive{\varX}{\pmvB}}$.
In $\confS{\runS}$ there exists a deposit $\confDep[\varX]{\pmvA}{\valV:\tokT}$.
$\labC = \pmvB \rightarrow *:m$, 
where $m$ is the signature of $\pmvA$ on a transaction $\txT$ satisfying the 
conditions in item~\ref{item:coherence:actGive}.
The mappings $\txMapC$, $\tkMapC$, $\burnMapC$ are equal to
$\txMapCi$, $\tkMapCi$, $\burnMapCi$, respectively.

\item \label{item:coherence:authBurn}
$\labS = \labAuth[x_j]{\pmvA[j]}{\actBurn{\vec{x}}{y}}$.
Let $\vec{\varX} = \varX[1] \cdots \varX[n]$.
In $\confS{\runS}$ there exist $n$ deposits $\confDep[{\varX[i]}]{\pmvA[i]}{\valV[i]:\tokT[i]}$.
$\labC = \pmvB \rightarrow *:m$, 
where $m$ is the signature of $\pmvA[j]$ on a transaction $\txT$ satisfying the 
conditions in item~\ref{item:coherence:actBurn}.
The mappings $\txMapC$ and $\tkMapC$ are equal to $\txMapCi$ and $\tkMapCi$, respectively.
If $y \in \dom(\burnMapCi)$, then $\burnMapCi(\varY) = \txT$, and $\burnMapC = \burnMapCi$.
Otherwise, $\burnMapC$ extends $\burnMapCi$ with the binding $\varY \mapsto \txT$.

\end{enumerate}

\paragraph*{Inductive case 2}

$\coher{\runS}{\runC\labC}{\txMapC}{\tkMapC}{\burnMapC}$ holds if
$\coher{\runS}{\runC}{\txMapC}{\tkMapC}{\burnMapC}$ and one of the following cases applies:
\begin{enumerate}[(1)]

\item $\labC = \txT$, where no input of $\txT$ belongs to $\ran{\txMapC}$.
  
\item $\labC = \pmvA \rightarrow *:m$, where $\labC$ does not correspond to any symbolic move,
  according to the inductive case 1.

\end{enumerate}

\section{Supplementary material for~\Cref{sec:computational-soundness}}
\label{sec:app-computational-soundness}


\begin{proofof}{lem:coher:s-to-c}
  By induction on \mbox{$\coherRel{\runS}{\runC}{}$}.
  The base case is trivial.
  Otherwise, there are the following two cases:
  \begin{itemize}

  \item For inductive case 1 of the definition of coherence, we have
    $\runS = \runSi\labS$ and $\runC = \runCi\labC$,
    with \mbox{$\coherRel{\runSi}{\runCi}{}$}.
    There are two subcases:
    \begin{itemize}
    \item $x$ occurs in $\runSi$:
      by the induction hypothesis, $\txMapCi(x)$ is an output with the
      intended fields.  We therefore only need to ensure that
      $\txMapC(x)=\txMapCi(x)$.  Inspecting how each $\labS$ is
      handled in the definition of coherence, we observe that this
      always holds except when $x$ no longer occurs in
      $\confG[\runS]$.  The latter case contradicts the hypothesis.
    \item $x$ does not occur in $\runSi$:
      Inspecting how each $\labS$ is handled in the definition of
      coherence, we observe that this is only possible when
      $\labC = \txT$ and $\txMapC(x)$ points to some output
      of $\txT$ which has the intended fields.
    \end{itemize}

  \item For inductive case 2 of the definition of coherence, we have
    $\runS = \runSi$ and $\runC = \runCi\labC$, with
    \mbox{$\coherRel{\runSi}{\runCi}{}$}.
    By the induction hypothesis, $\txMapCi(x)$ is an output with the
    intended fields. We therefore only need to ensure that
    $\txMapC(x)=\txMapCi(x)$, which holds since in inductive case 2
    $\txMapCi = \txMapC$.
  \end{itemize}
\end{proofof}

\begin{proofof}{lem:coher:c-to-s}
  By induction on \mbox{$\coherRel{\runS}{\runC}{}$}.
  The base case is trivial, since it only involves bitcoins.
  Otherwise, there are the following two cases:
  \begin{itemize}
  \item For inductive case 1 of the definition of coherence, we have
    $\runS = \runSi\labS$ and $\runC = \runCi\labC$,
    with \mbox{$\coherRel{\runSi}{\runCi}{}$}.
    There are two subcases:
    \begin{itemize}
    \item $(\txT, i)$ also occurs (unspent) in $\runCi$:
      We now prove the items of the lemma:
      \begin{itemize}
      \item %
        when $(\txT,i).\tkid$ does not correspond to any outputs, we
        apply the induction hypothesis to prove it is unspendable.

      \item %
        When its $\tkid$ corresponds to some mapped output but
        $(\txT,i) \notin \ran(\txMapC)$ (case 1), we also have
        $(\txT,i) \notin \ran(\txMapCi)$, since otherwise $(\txT, i)$
        would have been spent by $\labC$ (as can be verified
        inspecting all the possible cases in the definition of
        coherence).  Hence the induction hypothesis proves it
        unspendable.

      \item %
        When instead $(\txT,i) = \txMapC(y)$ (case 2), we also have
        $(\txT,i) = \txMapCi(y)$, since otherwise $(\txT, i)$ would
        have been spent by $\labC$ (as can be verified inspecting all
        the possible cases in the definition of coherence).  The
        thesis then follows from the induction hypothesis and the fact
        that $\labS$ can not spend $y$ (otherwise $(\txT,i)$ would
        also be spent).
      \end{itemize}
      
    \item $(\txT, i)$ does not occur in $\runCi$:
      We have $\labC = \txT$.
      Inspecting how each $\labS$ is handled in the definition of
      coherence, we observe that this case is only possible when
      $(\txT, i).\tkid$ is some output in $\ran(\txMapC)$, which
      belongs to the blockchain.
      Further, inspecting the same cases, we observe that $(\txT, i)$
      must also belong to $\ran(\txMapC)$, it is spendable, and that
      it corresponds to a suitable symbolic deposit in
      $\confG[\runS]$.
    \end{itemize}

  \item For inductive case 2 of the definition of coherence, we have
    $\runS = \runSi$ and $\runC = \runCi\labC$, with
    \mbox{$\coherRel{\runSi}{\runCi}{}$}.
    We also have $\txMapCi = \txMapC$ in this case.
    There are two subcases:
    \begin{itemize}
    \item $(\txT, i)$ also occurs (unspent) in $\runCi$:
      the thesis immediately follows from the induction hypothesis,
      which involves the same $\runS$ and $\txMapC$.
    \item $(\txT, i)$ does not occur in $\runCi$:
      by the definition of coherence, this is possible only when
      $\labC = \txT$ and no input of $\txT$ belongs to $\txMapC$.
      We now prove the items of the lemma:
      \begin{itemize}
      \item %
        when $(\txT,i).\tkid$ does not correspond to any output, we
        want to prove $(\txT,i)$ is unspendable. Indeed, the first
        part of the token script requires either that 
        $(\txT,i).\tkid$ is the parent output (contradiction) or that
        it has the same script and $\tkid$.
        In the latter case, by the induction hypothesis, the parent
        was unspendable -- contradiction.        

      \item %
        When $(\txT,i).\tkid$ corresponds to some mapped output but
        $(\txT,i) \notin \ran(\txMapC)$ (case 1), we also
        want to prove it is unspendable.

        Indeed, the first part of the token script requires either
        that $(\txT,i).\tkid$ is the parent output or that it has the
        same script and $\tkid$.
        In the former case, by hypothesis we get that $(\txT,i).\tkid$
        (being a parent) does not belong to $\ran(\txMapC)$, so it
        does not correspond to any mapped output -- contradiction.
        In the latter case, by the induction hypothesis, the parent
        was unspendable -- contradiction.

      \item %
        When instead $(\txT,i) = \txMapC(y)$ (case 2), we want to
        prove it is spendable and that it corresponds to a suitable
        symbolic deposit.  This case is actually impossible since
        $\txMapC(y) = \txMapCi(y)$ is an output which already occurs
        in $\runCi$, contradicting the hypothesis $\labC = \txT$.
      \end{itemize}
    \end{itemize}
  \end{itemize}
\end{proofof}

\begin{proofof}[Computational soundness]{th:computational-soundness}
  Assume that $\runC$ satisfies the hypotheses, 
  but there does not exist any $\runS$ which is coherent to $\runC$
  and conforming to the symbolic strategies. 
  Consider the longest prefix $\runCi$ of $\runC$ such that
  there exists some $\runSi$
  which is coherent (to $\runCi$) 
  and conforming (to the computational strategies). 
  We have that $\runC = \runCi \labC \, \runCii$.
  We now show that either 
  $\runCi\labC$ has a corresponding symbolic run $\runSi\runSii$ 
  which is coherent and conforming to the symbolic strategies
  (contradicting the maximality of $\runCi$), or
  the adversary succeeded in a signature forgery 
  (which can happen only with negligible probability).
  We proceed by cases on $\labC$:

  \paragraph*{Case 1 --- $\labC$ is a broadcast message}

  Let $\labC = \pmvA \rightarrow *:m$.
  Then, either one of the items~\ref{item:coherence:authGen} to~\ref{item:coherence:authBurn} 
  of the inductive case 1 applies, or the inductive case 2 applies 
  (since, by construction, the inductive case 2 catches all the other cases).
  In any case, we find a coherent extension $\runSii$ (possibly empty).
  Further, $\runSi \runSii$ conforms to the symbolic strategies:
  \begin{inlinelist}
  \item if the inductive case 2 was applied, then $\runSii$ is empty, and so conformance 
    follows from the induction hypothesis;
  \item if the inductive case 1 was applied, then the broadcast message $m$ is a signature
    of some user $\pmvB$.
    If $\pmvB$ is dishonest, then $\Adv$ can compute the signature:
    then, there exists a symbolic strategy of $\Adv$ which performs such move,
    and the resulting run conforms to the users' symbolic strategies.
    If $\pmvB$ is honest, either the signature was forged by $\Adv$ 
    (but only with a negligible probability), 
    or produced by $\pmvB$'s computational strategy.
    In the latter case, since the computational strategy of $\pmvB$ is derived from the symbolic one,
    $\pmvB$'s symbolic strategy outputs an authorization of $\pmvB$,
    which is exactly the label contained in $\runSii$.
    Therefore, $\runSi \runSii$ conforms to the symbolic strategies.
  \end{inlinelist}
  
  \paragraph*{Case 2 --- $\labC$ is a transaction}
  Let $\labC = \txT$. We consider the following subcases, according to the inputs of $\txT$:
  \begin{itemize}
    
  \item if none of the inputs of $\txT$ belongs to $\ran(\txMapC)$, 
    then we achieve coherence and conformance with $\runSii$ empty.
    
  \item otherwise, if \emph{all} the inputs of $\txT$ in $\ran(\txMapC)$ 
    correspond to symbolic $\BTC$ deposits, there are two subcases:
    \begin{itemize}
      
    \item $\txT$ is obtained by one of the items 
      from \ref{item:coherence:actGen}, \ref{item:coherence:actSplit}, \ref{item:coherence:actJoin}
      or \ref{item:coherence:actGive} of the inductive case 1 of coherence.
      Note that, by requirement (ii) of the definition of computational strategies,
      then all the witnesses in $\txT$ must have been broadcast in a previous computational step:
      by coherence, the corresponding authorizations have been performed in the symbolic run.
      Then, we can choose for $\runSii$ the symbolic action corresponding to $\txT$,
      which makes $\runSii$ a coherent extension of $\runSi$.
      Conformance holds trivially, by choosing $\Adv$'s symbolic strategy to perform 
      the symbolic action corresponding to $\txT$.
      
    \item $\txT$ is obtained by the item \ref{item:coherence:actBurn}
      of the inductive case 1 of coherence
      (by construction, this catches all the other cases).
      The proof proceeds as in the previous case.
      
    \end{itemize}
    
  \item otherwise, there exists at least one input of $\txT$ in $\ran(\txMapC)$ 
    corresponding to a deposit of a user-defined token, say $\confDep[\varX]{\pmvA}{\valV:\tokT}$.
    Let $(\txTi,i) = \txMapC(\varX)$: since this output stores a user-defined token,
    then it must be $(\txTi,i).\txscript = \tokScript$.
    Since $\txT$ spends $(\txTi,i)$, then $\tokScript$ evaluates to true.
    We proceed by cases on $(\txT,1).\tkop$:
    \begin{itemize}


    \item $\tkop = 1$. In this case appending $\txT$ corresponds to performing a $\burnOp$ action.
      The script $\tokScript$ in $(\txTi,i)$ ensures that
      $\txT$ has only one input, \ie $(\txTi,i)$, and one outputs.
      Since $\tokScript$ requires the signature of $(\txTi,i).\tkown$,
      and since by requirement (ii) of the definition of computational strategies,
      all the witnesses in $\txT$ must have been broadcast in a previous computational step,
      by coherence, it follows that the corresponding authorization (with a fresh name $\varY$)
      has been performed in the symbolic run.          
      Let $\runSii = \actBurn{\varX}{\varY}$.
      Then $\runSii$ is a valid extension of $\runSi$, since the 
      preconditions of rule $\burnRule$ are respected.
      Further, $\runSi\runSii$ is coherent with $\runCi \labC$ by the item \ref{item:coherence:actBurn}
      of the inductive case 1 of coherence.
      Conformance holds trivially, by choosing $\Adv$'s symbolic strategy to perform 
      the symbolic action corresponding to $\txT$ (this also applies to all subcases below).

    \item $\tkop = 2$. In this case appending $\txT$ corresponds to performing a $\splitOp$ action.
      The script $\tokScript$ in $(\txTi,i)$ ensures that:
      \begin{enumerate}
      \item $\txT$ has only one input, \ie $(\txTi,i)$, and two outputs.
        Let $\valV[1] = (\txT,1).\tkval$, let $\valV[2] = (\txT,2).\tkval$,
        and let $\pmvB = (\txT,2).\tkown$;
      \item $(\txT,1).\txscript = \tokScript$;
      \item $\valV[1],\valV[2] \geq 0$, and $\valV[1] + \valV[2] = \valV$.
      \end{enumerate}
      Since $\tokScript$ requires the signature of $(\txTi,i).\tkown = \pmvA$,
      and since by requirement (ii) of the definition of computational strategies,
      all the witnesses in $\txT$ must have been broadcast in a previous computational step,
      by coherence, it follows that the corresponding authorization has 
      been performed in the symbolic run.
      Let $\runSii = \actSplit{\varX}{\valV[1]}{\pmvB}$.
      Then $\runSii$ is a valid extension of $\runSi$, since the 
      preconditions of rule $\splitRule$ are respected.
      Further, $\runSi\runSii$ is coherent with $\runCi \labC$ by the item \ref{item:coherence:actSplit}
      of the inductive case 1 of coherence.

    \item $\tkop = 3$. 
      The script $\tokScript$ in $(\txTi,i)$ ensures that:
      \begin{enumerate}
      \item $\txT$ has two inputs and one output: let $(\txTii,j)$ be the other input, sibling of $(\txTi,i)$;
      \item $(\txT,1).\txscript = (\txTii,j).\txscript = \tokScript$;
      \item $(\txTi,i).\tkid = (\txT,1).\tkid = (\txTii,j).\tkid$.
      \end{enumerate}
      There are two further subcases, according to whether $(\txTii,j) \in \ran(\txMapC)$ or not:
      \begin{itemize}

      \item if $(\txTii,j) \in \ran(\txMapC)$, 
        then this input corresponds to a deposit $\confDep[\varY]{\pmvB}{\valVi:\tokT}$, and
        appending $\txT$ corresponds to performing a $\joinOp$ action.
        By item 3 above, both $\varX$ and $\varY$ store the same token $\tokT$.
        Since the script in $(\txTi,i)$ requires the signature of $\pmvA$,
        and that in $(\txTii,j)$ requires the signature of $\pmvB$,
        and since by requirement (ii) of the definition of computational strategies,
        all the witnesses in $\txT$ must have been broadcast in a previous computational step,
        by coherence, it follows that the corresponding authorizations have 
        been performed in the symbolic run.
        Let $\pmvC = (\txT,1).\tkown$.
        Let $\runSii = \actJoin{\varX}{\varY}{\pmvC}$. 
        Then $\runSii$ is a valid extension of $\runSi$, since the 
        preconditions of rule $\joinRule$ are respected.
        Further, $\runSi\runSii$ is coherent with $\runCi \labC$ by the item \ref{item:coherence:actJoin}
        of the inductive case 1 of coherence.

      \item if $(\txTii,j) \not\in \ran(\txMapC)$, 
        then by the first item of~\Cref{lem:coher:c-to-s}, the output $(\txTii,j)$ 
        is unspendable --- contradiction.
        Then, this case does not apply.

      \end{itemize}

    \item $\tkop = 4$. 
      The script $\tokScript$ in $(\txTi,i)$ ensures that:
      \begin{enumerate}
      \item $\txT$ has two inputs and two outputs: let $(\txTii,j)$ be the other input, sibling of $(\txTi,i)$;
      \item $(\txT,1).\txscript = \tokScript$;
      \end{enumerate}
      There are two further subcases:
      \begin{itemize}

      \item if $(\txTii,j) \not\in \ran(\txMapC)$, 
        then appending $\txT$ corresponds to performing a $\giveOp$ action.
        Since the script in $(\txTi,i)$ requires the signature of $\pmvA$,
        and since by requirement (ii) of the definition of computational strategies,
        all the witnesses in $\txT$ must have been broadcast in a previous computational step,
        by coherence, it follows that the corresponding authorization has
        been performed in the symbolic run.
        Let $\pmvB = (\txTii,j).\tkown$, and let $\runSii = \actGive{\varX}{\pmvB}$. 
        Then $\runSii$ is a valid extension of $\runSi$, since the 
        preconditions of rule $\giveRule$ are respected.
        There are two subcases.
        If $(\txTii,j) = \txT.\txIn[1]{}$, then 
        $\runSi\runSii$ is coherent with $\runCi \labC$ by the item \ref{item:coherence:actGive}c
        of the inductive case 1 of coherence.
        Otherwise, if $(\txTii,j) = \txT.\txIn[2]{}$ we obtain coherence 
        by the item \ref{item:coherence:actGive}b
        of the inductive case 1 of coherence.

      \item otherwise, if $(\txTii,j) \in \ran(\txMapC)$,
        then this input corresponds to a deposit $\confDep[\varY]{\pmvB}{\valVi:\tokTi}$.
        Appending $\txT$ corresponds to performing an $\exchangeOp$ action.
        Since the script in $(\txTi,i)$ requires the signature of $\pmvA$,
        and that in $(\txTii,j)$ requires the signature of $\pmvB$,
        and since by requirement (ii) of the definition of computational strategies,
        all the witnesses in $\txT$ must have been broadcast in a previous computational step,
        by coherence, it follows that the corresponding authorizations have 
        been performed in the symbolic run.
        There are two further subcases:
        \begin{itemize}
        \item $\tokTi = \BTC$. 
          Then, it must be 
          $(\txTi,i) = \txT.\txIn[1]{}$ and
          $(\txTii,j) = \txT.\txIn[2]{}$.
          Let $\runSii = \actExchange{\varX}{\varY}$. 
          
        \item if $\tokTi \neq \BTC$.
          If $(\txTi,i) = \txT.\txIn[1]{}$ and $(\txTii,j) = \txT.\txIn[2]{}$
          then let $\runSii = \actExchange{\varX}{\varY}$, 
          otherwise let $\runSii = \actExchange{\varY}{\varX}$. 
      
        \end{itemize}
        Then, $\runSi\runSii$ is coherent with $\runCi \labC$ by the item \ref{item:coherence:actExchange}
        of the inductive case 1 of coherence.
      \end{itemize}
      
    \item $\tkop = 5$. In this case appending $\txT$ corresponds to performing a $\giveOp$ action.
      The script $\tokScript$ in $(\txTi,i)$ ensures that:
      \begin{enumerate}
      \item $\txT$ has only one input, \ie $(\txTi,i)$, and one output;
      \item $(\txT,1).\txscript = \tokScript$;
      \item $(\txT,1).\tkval = (\txTi,i).\tkval$;
      \item $(\txT,1).\tkid = (\txTi,i).\tkid$.
      \end{enumerate}
      Since $\tokScript$ requires the signature of $(\txTi,i).\tkown = \pmvA$,
      and since by requirement (ii) of the definition of computational strategies,
      all the witnesses in $\txT$ must have been broadcast in a previous computational step,
      by coherence, it follows that the corresponding authorization has 
      been performed in the symbolic run.
      Let $\runSii = \actGive{\varX}{\pmvB}$.
      Then $\runSii$ is a valid extension of $\runSi$, since the 
      preconditions of rule $\giveRule$ are respected.
      Further, $\runSi\runSii$ is coherent with $\runCi \labC$ by the item \ref{item:coherence:actGive}a
      of the inductive case 1 of coherence.
      \qedhere


    \end{itemize}
    
  \end{itemize}
  
\end{proofof}

\bartnote{Nota per la completeness: sia $\txT$ una $\exchangeOp$ computazionale con input 
  $(\txTi,i)$ e $(\txTii,j)$ non mappato. Siccome $(\txTi,i)$ \`e mappato, l'esistenza della firma 
  computazionale implica quella della corrispondente autorizzazione simbolica. 
  Supponiamo che la firma su $\txT$ per l'input $(\txTii,j)$ non sia stata fatta. 
  Sotto queste ipotesi \`e possibile fare una $\giveOp$ simbolica, 
  che non ha una controparte computazionale. Infatti, non posso fare $\giveOp$ computazionale,
  perch\`e ho firmato una $\txT$ di $\exchangeOp$; non posso nemmeno fare la $\exchangeOp$, 
  perch\`e manca la firma. Per avere una bisimulazione, forse si potrebbe sdoppiare la regola $\giveRule$.
  La nuova $\giveRule$ richiede due autorizzazioni che corrispondono ai due witness di $\txT$.
  Il passo $\giveOp$ simbolico descritto sopra non si pu\`o pi\`u fare, perch\'e i passi computazionali 
  sono coerenti con la $\nrule{[Give2]}$, non con la $\giveRule$ classica.
}



}
{}

\end{document}